\documentclass[10pt,journal,compsoc, twoside]{IEEEtran}

\usepackage{tikz}

\usepackage{dirtytalk}

\usepackage{tcolorbox}
\usepackage{xcolor}
\usepackage{amsmath,bm}
\usepackage{amsthm}
\usepackage{amssymb}
\usepackage{graphicx} 
\graphicspath{{./figures_v2/}}
\usepackage{subfigure}
\usepackage{url}
\usepackage{booktabs} 
\usepackage{array} 
\usepackage{verbatim} 
\usepackage{caption}

\newcommand*{\permcomb}[4][0mu]{{{}^{#3}\mkern#1#2_{#4}}}
\newcommand*{\perm}[1][-3mu]{\permcomb[#1]{P}}

\usepackage{enumitem}
\usepackage{bbm}
\usepackage[font=footnotesize]{caption}

\newtheorem{theorem}{Theorem}

\newtheorem{definition}{Definition}
\newtheorem{lemma}{Lemma}

\newtheorem*{problem}{Optimization Problem}

\newtheorem*{intermediate}{Intermediate Step}
\usepackage{thmtools}
\declaretheoremstyle[headfont=\normalfont]{normalhead}

\newtheorem*{result}{Result}

\usepackage{url}
\usepackage{breakurl}

\usepackage{hyperref}
\usepackage{cleveref}
\makeatletter
\def\@ythm#1#2#3[#4]{\def\@currentlabelname{#4}%
  \expandafter\global\expandafter\def\csname#1name\endcsname{#4}%
  \@opargbegintheorem{#3}{\csname the#2\endcsname}{#4}%
  \ifx\thm@starredenv\@undefined
    \thm@thmcaption{#1}{{#3}{\csname the#2\endcsname}{#4}}\fi
  \ignorespaces}

\makeatother

\usepackage{color}

\newcommand{\eq}[1]{Eq.~\eqref{#1}}
\usepackage{soul}
\setstcolor{red}

\newcommand{\oldtext}[1]{}

\newcommand{\added}[1]{\textcolor[rgb]{0.00,0.00,0.00}{{#1}}}

\newcommand{\myitem}[1]{\vspace{0.25\baselineskip}\noindent\textbf{#1}}
\newcommand{\secref}[1]{Section~\ref{#1}}

\usepackage{comment}
\usepackage{algorithm}
\usepackage{algorithmicx}
\usepackage{algpseudocode}
\usepackage[export]{adjustbox}

\begin{document}


\title{Network Friendly Recommendations: Optimizing for Long Viewing Sessions} 

\author{Theodoros Giannakas,
Pavlos Sermpezis, and~Thrasyvoulos Spyropoulos
\IEEEcompsocitemizethanks{\IEEEcompsocthanksitem T. Giannakas, and T. Spyropoulos are with the Communication Systems
Department, EURECOM, France.\protect\\
E-mail: \{theodoros.giannakas, thrasyvoulos.spyropoulos\}@eurecom.fr
\IEEEcompsocthanksitem P. Sermpezis is with the Department of Informatics, Aristotle University of Thessaloniki, Greece.\protect\\
E-mail: sermpezis@csd.auth.gr}}


\IEEEtitleabstractindextext{%
\begin{abstract}
Caching algorithms try to predict content popularity, and place the content closer to the users. Additionally, nowadays requests are increasingly driven by recommendation systems (RS). These important trends, point to the following: \emph{make RSs favor locally cached content}, this way operators reduce network costs, and users get better streaming rates. Nevertheless, this process should preserve the quality of the recommendations (QoR). In this work, we propose a Markov Chain model for a stochastic, recommendation-driven \emph{sequence} of requests, and formulate the problem of selecting high quality recommendations that minimize the network cost \emph{in the long run}. While the original optimization problem is non-convex, it can be convexified through a series of transformations. Moreover, we extend our framework for users who show preference in some positions of the recommendations' list.
To our best knowledge, this is the first work to provide an optimal polynomial-time algorithm for these problems. Finally, testing our algorithms on real datasets suggests significant potential, e.g., $2\times$ improvement compared to baseline recommendations, and 80\% compared to a greedy network-friendly-RS (which optimizes the cost for I.I.D. requests), while preserving at least 90\% of the original QoR. Finally, we show that taking position preference into account leads to additional performance gains. 
\end{abstract}

}

\maketitle

\section{Introduction}
\label{sec:intro}
\subsection{Background}
The use of Content Distribution Networks (CDNs) has been common practice in the Internet~\cite{sitaraman2014}. At the same time, large content providers are starting to operate their own CDNs (e.g., Neflix Open Connect~\cite{netflixOpen}) by placing and operating smaller data centers inside a network operator, a trend that will continue in the context of mobile edge clouds. Interest in caching research has been revived in the context of Information-Centric Networks (ICNs), and more recently in wireless networks; there, a number of studies suggest to install tiny caches (e.g., hard drives) at every small-cell or femto-node~\cite{femto}, bringing ideas from hierarchical caching~\cite{borst2010} into the wireless domain. 

Caching techniques essentially try to predict what content users will probably request, and store it closer to the user. Storing content close to the users
, can (i) reduce the network cost to serve a request, and (ii) improve user experience (e.g., better playout quality). 
%
Nevertheless, the rapidly growing catalog sizes (both for professional and user-generated content), smaller sizes per cache (e.g., at femto-nodes
) compared to traditional CDNs, and volatility of user demand when considering smaller populations, make the task of caching algorithms increasingly challenging~\cite{Paschos-misconceptions,ElayoubiRoberts15}. 

To overcome such challenges, a radical approach has been recently proposed~\cite{chatzieleftheriou2019TMC,cache-centric-video-recommendation,giannakas2018show,munaro2015content,sermpezis2018soft,guo2017caching,liu2018learning}, based on the observation that user demand is increasingly driven today by recommendation systems (RSs) of popular applications (e.g., Netflix, YouTube). Instead of simply recommending \emph{interesting content}, recommendations could instead be ``nudged'' towards \emph{interesting content with low access cost} (e.g., locally cached)~\cite{chatzieleftheriou2019TMC,kastanakis-cabaret-mecomm}:
the recommendation quality remains high, and the new content will incur a smaller (network) cost, or even be accessible at better quality (e.g., HD), due to the lower latency~\cite{doan2018tracing}.
This approach is appealing, potentially presenting a win-win situation for all involved parties. It has, nevertheless, attracted some research interest only very recently, mostly in empirical studies~\cite{cache-centric-video-recommendation} or heuristic schemes~\cite{chatzieleftheriou2019TMC,giannakas2018show}.

\subsection{Motivation and Contributions}

\added{The main motivation for our paper, is that a recommendation \emph{we do now}, not only affects the user's next choice and the related network cost, but also subsequent choices and costs. However, the works in ~\cite{chatzieleftheriou2019TMC,cache-centric-video-recommendation,munaro2015content,sermpezis2018soft,guo2017caching,liu2018learning}, base their analysis on independent and identically distributed (I.I.D.) request patterns, ignoring the fact that a user's session often consists of consuming multiple contents in sequence (e.g., YouTube, Spotify). Thus, selecting recommendations towards network cost minimization for this sequential process is the main focus of our work.}

\added{We briefly present the problem our paper targets: A user starts a session in some multimedia (video, music, etc.) application, which is equipped with: (a) an RS that suggests a \say{related list} of $N$ recommended items; this list relates to the item \emph{just visited} (requested) by the user, and (b) a search bar that can be used for typing, and thus requesting \emph{any} content of the catalog. The user transits to the next content either from the RS list (potentially exhibiting some preference for the recommendations that are placed higher in the list), or from the search-bar.
Our goal is to design a methodology that returns the \emph{optimal} recommendation policy, which will simultaneously keep the user satisfied \emph{at every request} and minimize the total network cost incurred by her requests in this \emph{long session}.
}

Having established our goal, here we summarize the technical contributions of this paper:

\myitem{(i) Sequential request model.}
We propose an analytical framework based on absorbing Markov chain theory, to model a user accessing a long sequence of contents, driven by a RS (\secref{sec:problem_setup}, and \secref{sec:problem_formulation}). The sequential request model better fits real user behavior in a number of popular applications (e.g. YouTube, Vimeo, personalized radio) compared to IRM models used in previous work~\cite{chatzieleftheriou2019TMC,sermpezis2018soft}. 

\myitem{(ii) Problem formulation and convex equivalent.}
We formulate a generic optimization problem for high quality but network-friendly recommendations (we refer to these as ``Network-Friendly Recommendations'' or NFR). We show that this problem is non-convex, but we prove an equivalent convex one through a sequence of transformations (\secref{sec:optimization_methodology}). 


\myitem{(iii) Position Preference.}
We extend the established user model so that it takes into account the expressed user preference on some recommendation positions (\secref{sec:position}). We modify accordingly the optimization problem components, i.e., the variables and the constraints and show that the new one can also be transformed to a convex equivalent.

\myitem{(iv) Real-World Data Validation.}
We validate our algorithms using existing and collected datasets from different content catalogs, and demonstrate performance improvements up to $3 \times$ compared to baseline recommendations, and 80\% compared to a greedy cache-friendly recommender, for a scenario with 90\% of the original recommendation quality (\secref{sec:sims}). 

Finally, we discuss related work in \secref{sec:related} and present a set of open related problems in \secref{sec:discussion}.



\section{Problem Setup}
\label{sec:problem_setup}
\subsection{Problem Definitions}



We consider a user that consumes one or more contents during a session, drawn from a catalogue $\mathcal{K}$ of cardinality $K$.

\begin{definition}[Recommendation-Driven Requests]\label{def:requests} 
During the consumption of content $i \in \mathcal{K}$, a list of $N$ new contents are recommended to her, and she
\begin{itemize}[leftmargin=*,noitemsep,topsep=0pt]
\item follows recommendations with some fixed probability $\alpha\in(0,1)$ and picks uniformly among the $N$ contents. 
\item ignores the recommendations with probability $1-\alpha$, and picks a content $j$ (e.g., through a search bar) with probability $p_{0j} \in (0,1)$, $\mathbf{p}_{0} = [p_{01}, p_{02}, \dots, p_{0K}]^{T}$.
\end{itemize}
\end{definition}
\myitem{Assumptions on $\mathbf{p}_0$.}
For simplicity, we assume a type of time-scale separation is in place, where the probabilities $p_{0j}$ capture long-term user behavior (beyond one session).
W.l.o.g. we \emph{also} assume $\mathbf{p}_{0}$ governs the first content accessed, when a user starts a session. 

The above modeled session captures a number of everyday scenarios (e.g., watching clips on YouTube, personalized radio), where $\alpha$ captures the average probability of the user following recommendations (e.g., $\alpha = 0.5$ was measured for YouTube~\cite{RecImpact-IMC10}, and $0.8$ for Netflix~\cite{gomez2016netflix}, or $\alpha=1$ in the case of \textit{AutoPlay}). It is reported that YouTube users spend on average around 40 minutes at the service, viewing several related videos~\cite{businessYoutubeSessions}.

\myitem{Content Retrieval Cost.} We assume that fetching content $i$ is associated with a generic cost $c_{i}\in\mathbb{R}$, $\mathbf{c} = [c_{1}, c_{2},...,c_{K}]^{T}$, which is known to the content provider, and might depend on access latency, congestion overhead, popularity, file size, or even monetary cost. 

\noindent\emph{Minimizing cache misses:} Can be captured by setting $c_i = 0$ for all cached content and to $c_i = 1$, for non-cached content. 

\noindent\emph{Hierachical caching:} Can be captured by letting $c_i$ take values out of $m$ possible ones, corresponding to $m$ cache layers: higher values correspond to layers farther from the user~\cite{borst2010, poularakis2014toc}. 

\myitem{Remark:} While we have assumed, for simplicity, that these costs are associated with caching, this is not a requirement for our framework. Any network problem that gives us as input such cost values $c_i$ could be solved by the proposed approach. What is more, we are assuming that these costs (e.g. the contents cached) are fixed, at least during some time frame. This is inline with the standard femto-caching approach of ``cache today, consume tomorrow''~\cite{femto,wang2014cache,bastug2014living}, and the recent paradigm of \say{popular content prefetching} followed by Netflix~\cite{netflixOpen} and Google~\cite{googlepeer}. However, dynamic caching policies like LRU would require a different treatment (some details in Section~\ref{sec:discussion}). 



\begin{definition}[Matrix $\mathbf{U}$ - Content Relations]
For every pair of contents $i,j \in \mathcal{K}$, a score $u_{ij}\in[0,1]$ is calculated, using a state-of-the-art method.
\footnote{$u_{ij}$ could correspond to the \emph{cosine similarity} between content $i$ and $j$, in a collaborative filtering system~\cite{sarwar2001item}, or simply take values either $1$ (for a small number of related files) and $0$ (for unrelated ones). These scores might also depend on user preferences (e.g., past history).}
These values populate the square $K \times K$ matrix $\mathbf{U}$, which is assumed to be known to the RS.
\end{definition}

\begin{definition}[Control Variable $\mathbf{R}$]\label{def:control-variables}
Let $r_{ij} \in [0,1]$ denote the probability that content $j$ is recommended after a user watches content $i$. These probabilities define a square $K \times K$ recommendation matrix $\mathbf{R}$, over which we optimize.  
\end{definition}

\myitem{Baseline Recommendations.} Recommendation systems (RS) is an active area of research, with state-of-the-art RSs using collaborative filtering~\cite{sarwar2001item}, matrix factorization~\cite{koren2009matrix}, deep neural networks~\cite{covington2016deep} and recently Q-Learning~\cite{slateQ}.
%
For simplicity, we assume that the baseline RS works as follows:
\begin{definition}[Baseline Recommendations]\label{def:baseline-RS}
For every content $i \in \mathcal{K}$, the baseline RS at content $i$ will always recommend the $N$ items\footnote{$N$ depends on the scenario. E.g., in YouTube $N=1$ when \textit{autoplay} mode is on, $N=2,..,5$ in its mobile app, and $N=20$ in its website version.} with the highest $u_{ij}$ values~\cite{RecImpact-IMC10}. In other words, $\mathbf{r}_{i}^{base}$, the $i$-th row of $\mathbf{R}^{base}$, will be a vector of size $K$, indexed by $N$ 1's at the position of the highest $u_{ij}$. Thus, for every content $i$, the baseline RS achieves:
\begin{equation}\label{eq:qmax}
    q_{i}^{max} = \sum_{j=1}^K r_{ij}^{base} \cdot u_{ij}
\end{equation}
\end{definition}

Our goal is to design a policy $\mathbf{R}$, which is different from $\mathbf{R}^{base}$ (see Def.~\ref{def:baseline-RS}); $\mathbf{R}^{base}$ is based only on $\mathbf{U}$, and satisfies the users by offering $q^{max}$, whereas we are interested in designing an RS that considers (a) $\mathbf{U}$ and (b) access costs $\mathbf{c}_i$ of all contents available in the library, and satisfy also the network needs. As a warm-up, when we formulate our problem in the next section: we will be interested in minimizing a criterion that is based on the access cost, by guaranteeing some level of quality of recommendations (QoR) and treat it as a constraint.

\begin{definition}[Network-friendly RS] 
The $q$-Network-friendly RS is the one that achieves at least $q$, with $q\in [0, 1]$, of the $q_{i}^{max}$ for every content $i$.
Therefore, the set of $q$-Network-friendly RS, is the RSs that obey the following set of $K$ inequality constraints.

\begin{equation}\label{eq:quality}
\sum_{i=1}^{K} r_{ij} \cdot u_{ij} \ge q \cdot q_{i}^{max}, \forall i\in\mathcal{K}.
\end{equation}
\end{definition}

Our $q$-NFRS with policy $r_{ij}$, will guarantee the quality of recommendations (QoR) through the set of constraints Eq. (\ref{eq:quality}) (the achieved QoR is on the left handside of this expression), where $q$ -a tuning parameter of the RS- decides the percentage of the $q^{max}_{i}$ quality we offer. Importantly, when $q\to 0$, QoR is low and the RS recommends based only on the access cost (opportunity for large network gains), whereas if $q\to 1$, the RS becomes $\mathbf{R}^{base}$ (the optimization problem is \say{very} constrained) and the RS cannot improve network access cost. In this paper, we will focus on values of $q>70\%$ in order to capture interesting scenarios and see whether low network access cost can be achieved by keeping the users happy at the same time.

\myitem{Remark:}
\added{In addition to network delivery cost, network-friendly recommendations might also improve user QoE: for example, a locally cached content could be fetched more efficiently (lower latency, higher bandwidth, etc.) and streamed without interruptions in High Definition, an obvious ``win-win'' situation for the network operator and the users. Recent experimental studies provide evidence and quantify such 
QoE improvements~\cite{doan2018tracing,sermpezis2019towards}. In this context, optimization-wise, there are other interesting choices for jointly modeling the QoR and QoE. For instance, a way to capture user satisfaction (in this paper captured only through QoR), would be to add a second constraint that relates only to QoE-related metrics~\cite{sermpezis2019towards}.}


\subsection{Examples}\label{subsec:examples}


\myitem{Probabilistic Recommendations.}
The probabilistic way of defining recommendations enables us to capture generic scenarios. Consider a library of size $K=5$, and an application requiring $N=2$ recommended items. Assume that a user currently consumes content 1, and let the first row of the matrix $\mathbf{R}$ to be $\mathbf{r}_{1|\cdot} = [0.0, 1.0, 0.5, 0.5, 0.0]$. In practice, this means that after consuming content 1, content 2 will always be recommended, and the second recommendation will be for content 3 or 4 with equal probability ($r_{13}=r_{14}=0.5$).

\myitem{Increasing Hit Rate now.}
To exemplify the $q$-NFRS concept, assume again $K=5$, content 5 is cached, $q=0.8$ and $\mathbf{u}_{1|\cdot} = [0, 1.0, 1.0,0.2,0.0]$, and $q_{1}^{max} = 2.0$, see Eq.(\ref{eq:qmax}). A $q$-NFRS with interest in maximizing its cache hit would have the following policy in content 1, that is $\mathbf{r}_{1|\cdot} = [0.0, 0.8, 0.8, 0.0, 0.4]$. This way, it would satisfy the constraint but at the same time drive the user also towards item 5, which is cached.

\myitem{Increasing Hit Rate for the future.}
The example in Fig.~\ref{fig:example-comparison-rec} depicts such a scenario, where the user consumes 5 items in sequence. The RS on the left suggests the most relevant item to the currently viewed \emph{all the time}; this results in a hit rate of 20$\%$. Interestingly, on the same figure on the right, we see the RS arranging a non-trivial policy: It offers the most relevant item at all times except when it finds the user at item 3, where it slightly degrades the quality of recommendations ($u_{34}=0.8$). This simple move however, drastically changes the path of requested contents and increases the hit rate in the long run from 20$\%$ to 60$\%$.

\begin{figure}
\centering
\subfigure{\includegraphics[width=0.8\columnwidth]{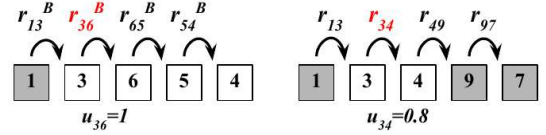}}
\vspace{-2.5mm}
\caption{Comparison of baseline (left) and network-friendly (right) recommenders. Gray and white boxes denote cached and non-cached contents, respectively. Recommending after content $3$ a slightly less similar content (i.e., content $4$ instead of $6$), leads to lower access cost in the long term.}
\label{fig:example-comparison-rec}
\end{figure}

Table~\ref{table:notation} summarizes some important notation. Vectors and matrices are denoted with bold symbols.

\vspace{5pt}
\begin{table}[h]
\centering
\caption{Important Notation}\label{table:notation}
\begin{small}
\begin{tabular}{|l|l|}
\hline
{$\alpha$}			&{Prob. the user follows recommendations}\\
\hline
{$r_{ij}$}			&{Prob. to recommend $j$ after viewing $i$}\\
\hline
{$q_{i}^{max}$}		&{Maximum baseline quality of content $i$}\\
\hline
{$q$}		        &{Percentage of original quality}\\
\hline
{$\mathbf{p}_0$}	&{Baseline popularity of contents}\\
\hline
{$u_{ij}$}			&{Similarity scores content pairs $\{i,j\}$}\\
\hline
{$c_i$}			    &{Access cost for content i}\\
\hline
{$\mathcal{K}$}		&{Content catalogue (of cardinality $K$)}\\
\hline
{$N$}				&{Number of recommendations}\\
\hline
\end{tabular}
\end{small}
\end{table}

\section{Problem Formulation}
\label{sec:problem_formulation}
The goal of this paper is to carefully select recommendations in order to reduce the content access cost for users that have long sessions in multimedia applications. In this section, we initially cast the user request process as an Absorbing Markov Chain (AMC) (Section \ref{sec:RR}), which then helps us to derive the expected content access cost for a user session (Section~\ref{sec:expected-cost}). Finally we conclude the section by formulating the optimization problem of network-friendly recommendations (Section~\ref{sec:optim-problem-formulation}).

\subsection{Renewal Reward Process}\label{sec:RR}
As we described earlier, a session for a recommendation-driven user consists of a sequence of periods during which she follows recommendations, say $S_R$, intermixed with steps at which the user ignores recommendations (see Def. \ref{def:requests}). To better visualize such a session see Fig.~\ref{fig:example2}. We will use the following two arguments to model such a session:
(i) Each $S_R$ period can be modeled with an absorbing Markov chain with transition matrix $\mathbf{P}$ (show matrix) of size $K+1 \times K+1$,

\begin{equation}\label{eq:transition-matrix}
\mathbf{P} = \left(
\begin{array}{@{}c|c@{}}
\frac{\alpha}{N} \cdot \mathbf{R} & \begin{array}{@{}c@{}} 1 - \alpha \\ \vdots \\ 1 - \alpha \end{array} \\
\cline{1-1}
\multicolumn{1}{@{}c}{
  \begin{matrix} 0 & \cdots & 0 \end{matrix}
} & 1
\end{array}
\right)
\end{equation}

where the transient part $\mathbf{Q} = \frac{\alpha}{N}\cdot \mathbf{R}$ corresponds to the user following recommendations (according to our control $K \times K$ variable $\mathbf{R}$); each such period can end at any step with a probability $1-\alpha$, modeled as an additional absorbing state.
(ii) When a recommendation period ends, the process gets ``renewed'', that is since the user ''re-enters'' the catalog from the same initial distribution $\mathbf{p}_0$ when not following recommendations, each $S_R$ period is I.I.D. 

Hence, a user session can be modeled as a renewal process, that renews after each recommendation period (i.e. every time the user decides to not follow recommendations). In the following, we use the above AMC to derive the expected cost per recommendation period, and the renewal reward theorem to derive the expected cost of the entire session, which will serve as our optimization problem objective.

\begin{definition}[Content Sequence]\label{lemma:RR}
A content access sequence $S = \{S_{R}^{1},S_{R}^{2},\dots\}$ defines a renewal process, with subsequences $S_{R}^{i}$, where the user follows recommended content, each ending with a jump outside of the RS. The cumulative cost of contents $C(S_{R}^{i})$ that incurred during a cycle $S_{R}^{i}$ is the cost of that cycle.
\end{definition}


\begin{figure}
\centering
\subfigure{\includegraphics[width=0.7\columnwidth]{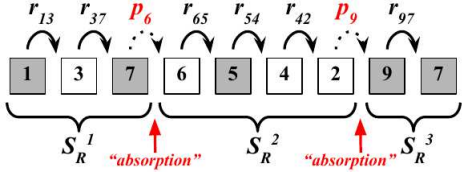}}
\vspace{-2mm}
\caption{Example of a multi-content session. Gray and white boxes denote cached and non-cached contents, respectively. A user follows recommendations (continuous arrows) or ignores them (dotted arrows).}
\label{fig:example2}
\end{figure}

%

\subsection{Long Term Expected Cost} \label{sec:expected-cost}

The goal of this subsection is to derive the long term expected cost of a user session. If we denote as $C(S_{R}^{i})$ the cost of the $i$-th cycle, we get the following expression

\begin{align}
    E[C(S_{R})] = \lim_{T \to \infty} \frac{C(S_{R}^{1}) + \dots + C(S_{R}^{T})}{T}
\end{align}


\begin{lemma}[Recommendation-Driven Cost] \label{lemma:rec-cost}
The content access cost $C(S_{R})$ during a (recommendation) renewal cycle $S_{R}$ is given by
\begin{equation}\label{eq:cost-cycle}
E[C(S_{R})] = \mathbf{p}_0^{T} \cdot \mathbf{G} \cdot \mathbf{c},
\end{equation}

If we further denote as $|S_{R}|$ the expected length of such a cycle, then by using the geometric r. v. argument, the expected length is

\begin{equation}\label{eq:cycle-length}
E[|S_{R}|] = \mathbf{p}_0^{T} \cdot \mathbf{G} \cdot \mathbf{1} = \frac{1}{1-\alpha},
\end{equation}
where $\mathbf{G} = \left(\mathbf{I}-\frac{\alpha}{N} \cdot \mathbf{R}\right)^{-1}$ is the \say{fundamental matrix} of the AMC described by Eq. (\ref{eq:transition-matrix}). 
\end{lemma}

\begin{proof} 
Can be found in the Appendix.
\end{proof}

Finally, the following theorem which gives the long term expected cost, follows immediately from Def. \ref{lemma:RR}, Lemma~\ref{lemma:rec-cost}, and the Renewal-Reward theorem~\cite{mor2013}.
\begin{theorem}\label{thm:total-expected-cost}
The LTEC, for a long user session S, given a recommendation matrix $\mathbf{R}$ is

\begin{equation}\label{eq:cost-rate}
\lim_{T \to \infty} \sum_{i = 1}^T \frac{C(S_{R}^{i})}{T} \stackrel{RR} = \frac{E[C(S_R)]}{E[|S_R|]} =  
\frac{\mathbf{p}_0^{T} \cdot \left(\mathbf{I}-\frac{\alpha}{N} \cdot \mathbf{R}\right)^{-1} \cdot \mathbf{c}}{\frac{1}{1-\alpha}}.
\end{equation}
\end{theorem}

\subsection{Optimization Problem}\label{sec:optim-problem-formulation}
In \textbf{\nameref{problem:basis}} we formulate the optimization problem, where the goal is to minimize the expected cost given in Theorem~\ref{thm:total-expected-cost} (objective function), by selecting the recommendations $\mathbf{R}$ (optimization variables). 

\begin{problem}[OP-Uni]\label{problem:basis}
\begin{subequations}\label{eq:objective-infinite-step}
\begin{align}
\underset{\mathbf{R}}{\textnormal{minimize}}~~~ &
  \frac{\mathbf{p}_0^{T} \cdot \left(\mathbf{I} - \frac{\alpha}{N} \cdot \mathbf{R}\right)^{-1} \cdot \mathbf{c}}{\frac{1}{1-\alpha}}, \label{eq:objective-infinite-step}\\
\textnormal{subject to}~~~& \sum_{j = 1}^{K} r_{ij} \cdot u_{ij} \geq q \cdot q_i^{max}, ~~\forall i~\in \mathcal{K},
\label{quality-con}\\
& \sum_{j = 1}^{K}   r_{ij} = N, ~~\forall i~\in \mathcal{K} \label{affine-con}\\
& 0 \le r_{ij} \le 1 ~ (i \ne j), ~~ r_{ii} = 0. \label{box-con}
\end{align}
\end{subequations}
\end{problem}

As discussed earlier, recommendations need to satisfy the quality constraints of Eq. (\ref{eq:quality}) (captured in \eq{quality-con}), be exactly $N$ for each content (captured in \eq{affine-con}), and conform to Def.~\ref{def:control-variables} (captured in \eq{box-con}).

\section{Optimization Methodology}
\label{sec:optimization_methodology}
In this section, we deal with \textbf{\nameref{problem:basis}}, by first characterizing its convexity properties and then by applying a series of transformations that lead to a Linear Programming formulation.

\begin{lemma}\label{lemma:initial-nonconvex}
The problem described in \textbf{\nameref{problem:basis}} is nonconvex. 
\end{lemma}

\begin{proof}
The problem \textbf{\nameref{problem:basis}} comprises $K^{2}$ variables $r_{ij}$, and a set of $K^{2} + 2\cdot K$ linear (equality and inequality) constraints, thus the feasible solution space is convex. However, assume w.l.o.g that $\mathbf{p}_0 = \mathbf{c} = \mathbf{w}$; the objective now becomes $f(\mathbf{R})=\mathbf{w}^{T}(\mathbf{I}-\frac{a}{N} \cdot \mathbf{R})^{-1}\mathbf{w}$. Unless we constrain $\mathbf{R}$ to be in the class of symmetric and positive semidefinite matrices, the objective is nonconvex~\cite{boyd2004convex,ermon2014designing}.
\end{proof}

Hence, there is no polynomial time algorithm solving problem \textbf{\nameref{problem:basis}}.
While one might be tempted to reduce the feasible solution space of variable $\mathbf{R}$ and force it to be symmetric and positive semidefinite, and solve the problem as a convex SDP, this fundamentally leads to suboptimal solutions.



\subsection{Road to the Optimal Solution}\label{sec:journey-optimality}

A fundamental difficulty of \textbf{\nameref{problem:basis}} is the inverse matrix in the objective $\mathbf{p}_0^{T} \cdot \left(\mathbf{I} - \frac{\alpha}{N} \cdot \mathbf{R}\right)^{-1} \cdot \mathbf{c}$. 
A reasonable first action is to introduce $K$ auxiliary variables and set them equal to $\mathbf{z}^{T} = \mathbf{p}_0^T \cdot(\mathbf{I} - \frac{a}{N} \cdot \mathbf{R})^{-1}$. Multiplying both sides from the right with $(\mathbf{I} - \frac{a}{N} \cdot \mathbf{R})$ yields

\begin{align}\label{eq:new-equality-cons}
    \mathbf{z}^{T} \cdot (\mathbf{I} - \frac{a}{N} \cdot \mathbf{R}) = \mathbf{p}_0^T
\end{align} 

Hence, problem \textbf{\nameref{problem:basis}} is equivalent to the following \footnote{Two problems are equivalent if the solution of the one, can be uniquely obtained through the solution of the other; introducing auxiliary variables preserves the property. We refer the reader to~\cite{boyd2004convex} for more details.}

\begin{intermediate}[Equivalent formulation]\label{problem:eq-1}
\begin{subequations}
\begin{align}
\underset{\mathbf{z},~\mathbf{R} }{\textnormal{minimize}}~~~& 
  \mathbf{c}^T \cdot \mathbf{z},\label{eq:objective-interm1}\\
\textnormal{subject to}~~~&  \mathbf{z}^T - \frac{a}{N}\cdot \mathbf{z}^T\cdot \mathbf{R} = (1-\alpha) \cdot \mathbf{p}_0^T \label{stationarity-con}\\
& \sum_{j =1}^{K} r_{ij} \cdot u_{ij} \geq q \cdot q_i^{max},~~\forall i~\in \mathcal{K},\\
& \sum_{j = 1}^{K} r_{ij} = N,~~\forall i~\in \mathcal{K},\\
& 0 \le r_{ij} \le 1 ~ (i \ne j), ~~ r_{ii} = 0
\end{align}
\end{subequations}
\end{intermediate}

Observe that although now the objective is linear in the variable $\mathbf{z}$, the constraint of Eq. (\ref{stationarity-con}) is quadratic in the $\mathbf{z,R}$. There this step does not seem \emph{yet} like much of a progress.

\myitem{Discussion:} The above formulation falls under the umbrella of non-convex quadratically constrained quadratic program (QCQP), where it is common to perform a convex relaxation of the quadratic constraints, and then solve an approximate convex problem (e.g., semidefinite program (SDP) or Spectral relaxation, see~\cite{park2017general} for more details). 
The problem can also be seen as \emph{bi-convex} in variables $\mathbf{R}$ and $\mathbf{z}$, respectively. Alternating direction method of multipliers (ADMM) can be applied to such problems, iteratively solving convex subproblems~\cite{boyd2011distributed,giannakas2018show}. Nevertheless, none of these methods provides any optimality guarantees, and even convergence for non-convex ADMM is an open research topic~\cite{wang2015global,gao2019admm}.

The above discussion motivates us to pay closer attention to the problem structure and the actual meaning of the variables at hand. For this reason instead, we introduce an additional variable transformation, where we define variables $f_{ij}$ defined as $f_{ij} = z_{i} \cdot r_{ij}$. Focusing on the Eq. (\ref{stationarity-con}) in scalar form we have:
\begin{align}
    z_j & = \frac{\alpha}{N} \sum_{i=1}^K z_i \cdot r_{ij} + (1-\alpha) \cdot p_{0j} \Rightarrow \nonumber \\
    z_j & = \frac{\alpha}{N}\sum_{i=1}^K f_{ij} + (1-\alpha) \cdot p_{0j}
\end{align}
The new variables are $\mathbf{z}$ (vector of size $K\times 1$ vector) and $\mathbf{F}$ (matrix of size $K\times K$ matrix). 

\myitem{Interpretation of $f_{ij}$.}
\added{The vector appearing in the objective Eq. (\ref{eq:objective-infinite-step}) represents the stationary distribution of a PageRank-like model defined by our stochastic process~\cite{giannakas2018show,avrachenkov2006pagerank}.
Therefore, the scalar quantity $(1-\alpha) \cdot z_i$ expresses the long-term probability that item $i$ is requested. Given that, it is easy to see that $ z_i = \pi_{i} \cdot \frac{1}{1-\alpha}$ and as a consequence
\begin{align}
f_{ij} = \frac{1}{1-\alpha}\cdot \pi_{i}\cdot r_{ij}
\end{align}
Recall from our definitions: $r_{ij}$ denotes the probability to recommend content $j$ conditioned on the fact that the user is at content $i$; $f_{ij}$ translates to 
the percentage of time (in the long run) the user \emph{was} at $i$ \emph{and saw} $j$ in her RS list, scaled by the quantity $\frac{1}{1-\alpha}$.}
\begin{problem}[LP~(\textbf{\nameref{problem:basis}})]\label{problem:basis-LP}
\begin{subequations}\label{eq:objective-LP}
\begin{align}
\underset{\mathbf{z},~\mathbf{F}}{\textnormal{minimize}}~~~& 
  \mathbf{c}^T \cdot \mathbf{z}, \label{eq:objective-LP} \\
\textnormal{subject to}~~~& \sum_{j=1}^{K} f_{ij} \cdot u_{ij} - z_i \cdot q\cdot q_i^{max} \geq 0, ~\forall~i~\in\mathcal{K}
\label{quality-con-LP}\\
& \sum_{j = 1}^{K} f_{ij} - N\cdot z_i = 0,~\forall~i~\in\mathcal{K} \label{affine-con-LP}\\
& f_{ij} - z_i \le 0~\forall~i,j\in\mathcal{K} \label{affine-f-z-Ineq-LP}\\
& f_{ij}\ge 0~(i \neq j),~f_{ii} = 0 \label{f-positive-con-LP}\\
& z_j - \frac{\alpha}{N} \cdot \sum_{i=1}^K f_{ij} = p_{0j}, ~\forall j\in\mathcal{K} \label{zf-relaxation}
\end{align}
\end{subequations}
\end{problem}
For the set of constraints in \textbf{\nameref{problem:basis}} we simply substituted $r_{ij} = z_{i} \cdot f_{ij}$ and \eq{quality-con} $\to$ \eq{quality-con-LP}, \eq{affine-con} $\to$ \eq{affine-con-LP} and finally the left handside of \eq{box-con} $\to$ \eq{f-positive-con-LP} whereas its right handside becomes \eq{affine-f-z-Ineq-LP}. Finally, $r_{ii} = 0 \to f_{ii} = 0$.

\begin{lemma}\label{lemma:bijection}
The change of variables $f_{ij} = z_{i} \cdot r_{ij}$, is a one-to-one mapping between $(z_i,r_{ij})$ and $(z_i,f_{ij})$. 
\end{lemma}

\begin{proof}
To obtain $r_{ij}$, one needs to compute  $f_{ij}/z_{i}$ from the pair $\{z_i, r_{ij}\}$. In order to retrieve $r_{ij}$ from the above computation, the value $z_{i}$ should be strictly nonzero, as then $r_{ij}$ would be undefined. However, observe from 
\begin{align}
    z_{j} = \frac{\alpha}{N} \cdot \sum_{i=1}^{K} f_{ij} + p_{0j} > 0~\forall~j~\in~\mathcal{K}
\end{align}

since $f_{ij} \ge 0$ and $p_{0i} > 0, \forall i$ (see Def.~\ref{def:requests}), this forces $\mathbf{z}$ to be strictly positive and thus never zero. Therefore $r_{ij}$ are always uniquely defined provided that $p_{0i} > 0~\forall~i~\in~\mathcal{K}$.
\end{proof}

Note that the only condition we needed to establish in order for the Lemma \ref{lemma:bijection} to hold, is that \emph{all} contents must have a nonzero probability to be requested from the user.
Combining Lemma~\ref{lemma:bijection}, along with the definition of problem in the Intermediate Step, yields the final formulation \textbf{\nameref{problem:basis-LP}}.

\subsection{Computational benefits of equivalent problem}
The \nameref{problem:basis-LP} corresponds to a Linear Program and it consists of $2K^2 + 4K +1$ linear constraints. 
To combat LPs, there are plently of implemented widely used solvers (e.g., ILOG CPLEX, GUROBI, MOSEK).
There are two benefits in transforming our problem to an LP compared to the heuristic ADMM we presented in \cite{giannakas2018show}.

\begin{enumerate}
    \item Optimality guarantees for \nameref{problem:basis-LP}.
    \item No need for parameter tuning.
\end{enumerate}

Solving \textbf{\nameref{problem:basis}} via ADMM: (1) returns in principle a suboptimal solution, and (2) its performance heavily depends on carefully selecting the parameter $\mu$ (the penalty on the quadratic term), while the CPLEX has no need of tuning.
To validate this, we increase the library size $K$ and solve the \emph{exact same} instances of \textbf{\nameref{problem:basis}}, and we report the results (execution time, and cache hit rate) in Tables \ref{table:admm-LP-1} and \ref{table:admm-LP-2}; the details of the problem parameters are provided in Section~\ref{sec:sims}.


For the ADMM implementation \cite{giannakas2018show}, the inner minimization loops of the ADMM were implemented using cvxpy~\cite{cvxpy} and more specifically the solver SCS~\cite{o2016scs}. In the following simulations we chose the ADMM tuning parameters as $\mu = 30$. These experiments were carried out using a PC with RAM: 8 GB 1600 MHz DDR3 and Processor: 1,6 GHz Dual-Core Intel Core i5.
In both tables, the execution times of the LP-based solution is lower and returns a better objective value. However note that as we tighten the accuracy of the ADMM loop in Table \ref{table:admm-LP-2}, the suboptimality gap becomes much smaller, but that comes in cost of significantly higher execution times.

\begin{table}[h]
\centering
\caption{LP-based vs ADMM - Looser accuracy ($\epsilon=0.01$)} \label{table:admm-LP-1}
\begin{tabular}{l|c|c|c|c}
{Metrics $\to$} & \multicolumn{2}{c|}{Cache Hit Rate ($\%$)} & \multicolumn{2}{c}{Execution Time ($s$)} \\
\hline
{Method $\to$}               			& {LP} & {ADMM} & {LP} & {ADMM}  \\
\hline 
{$K = 50$}			    			    &{49.79} &{47.45} &{0.061} &{1.293} \\
\hline
{$K = 100$}			    			    &{53.40} &{51.86} &{0.571} &{8.560} \\
\hline
{$K = 150$}		        				&{49.74} &{46.23} &{1.804} &{16.401} \\
\hline
{$K = 200$}		        				&{52.34} &{47.39} &{6.560} &{113.44} \\
\hline
{$K = 250$}		        				&{51.74} &{46.66} &{9.384} &{162.154} \\
\hline
{$K = 300$}		        				&{52.01} &{49.29} &{15.534} &{464.787} \\
\end{tabular}
\end{table}

\begin{table}[h]
\centering
\caption{LP-based vs ADMM - Tighter accuracy ($\epsilon=0.001$)} \label{table:admm-LP-2}
\begin{tabular}{l|c|c|c|c}
{Metrics $\to$} & \multicolumn{2}{c|}{Cache Hit Rate ($\%$)} & \multicolumn{2}{c}{Execution Time ($s$)} \\
\hline
{Method $\to$}               			& {LP} & {ADMM} & {LP} & {ADMM}  \\
\hline
{$K = 50$}			    			    &{43.02} &{40.78} &{0.0605} &{1.825} \\
\hline
{$K = 100$}			    			    &{46.30} &{46.28} &{0.7802} &{18.476} \\
\hline
{$K = 150$}		        				&{50.62} &{50.32} &{2.437} &{124.14} \\
\hline
{$K = 200$}		        				&{50.27} &{49.23} &{5.773} &{457.96} \\
\end{tabular}
\end{table}



\subsection{Greedy Baseline Scheme}

Here, we will formalize a probabilistic but myopic RS that aims at minimizing the access cost of \emph{only the next immediate request}. Importantly, this will serve later as a heuristic baseline RS in the evaluation section.
Notably, the \emph{Greedy Baseline} approach resembles the policies proposed in \cite{chatzieleftheriou2019TMC,cache-centric-video-recommendation}.
More specifically, the algorithm of~\cite{chatzieleftheriou2019TMC} targets a different context (i.e., caching \emph{and} single access content recommendation); the Greedy method could be interpreted as applying the recommendation part of~\cite{chatzieleftheriou2019TMC} for each user, along with a continuous relaxation of the control (recommendation) variables. 


The approach we are following takes into account the dependence of actions in consecutive steps of the user, and attempts to minimize the long term cost $(\mathbf{P} + \mathbf{P}^2 + \dots) \cdot \mathbf{c}$, which we approximated with the stationary cost, see Eq. (\ref{eq:objective-infinite-step}). Hence, a simple approximation would be to keep only the first term of this expansion.
\begin{align}
    f(\mathbf{R}) = \mathbf{p}_1^T \cdot \mathbf{c} =  \bigg(\mathbf{p}_{0}^T \cdot \mathbf{P} \bigg) \cdot \mathbf{c}.
\end{align}
Optimizing this corresponds to a greedy (or ``myopic'') algorithm that tries to minimize the cost of the next step only. This gives rise to the following, simpler optimization problem:
\begin{problem}[OP-greedy-uni]\label{problem:myopic}
\begin{align}
\underset{\mathbf{R}}{\textnormal{minimize}}~~~ &
  \mathbf{p}_{0}^T \cdot \mathbf{R} \cdot \mathbf{c}, \label{eq:myopic}\\
\textnormal{subject to}~~~& \textnormal{Eqs. (\ref{quality-con}, \ref{affine-con}, \ref{box-con})} \label{constraints:myopic}
\end{align}
\end{problem}
The problem is an LP, as the objective can be readily written as $\sum_{i,j = 1}^K r_{ij} \cdot p_{0i}\cdot c_j$ and the set of constraints Eq. (\ref{constraints:myopic}) is the same convex set as \textbf{\nameref{problem:basis}} without the demanding set of constraints (\ref{stationarity-con}).

\myitem{Remark:} The objective can be split in to $K$ different summations, where each summand is independent. Moreover the constraints over each row of $\mathbf{R}$ are also independent. Hence, the problem is naturally decomposed into $K$ different LP.

\section{Non-Uniform Click-Through}
\label{sec:position}
In the previous sections, we have established that the problem of $q$-NFRS can be optimally solved for a recommendation-driven user defined as in Section~\ref{sec:problem_formulation}. A number of possible extensions to this simple model can be considered towards making it more realistic. We consider such an extension in this section.
Specifically, recent studies studies~\cite{what-should-you-cache-nossdav,RecImpact-IMC10}, have shown that users have the tendency to click on contents (or products in the case of e-commerce) according to their position in the list of recommendations.
Hence, the probability of picking content in the first position ($v_1$), may be higher than the probability to pick the content in position $N$ ($v_N$). In fact, a Zipf-like relation has been observed~\cite{RecImpact-IMC10}. 



\myitem{Assumption on Position Preference.}
The user behaves as in Def. \ref{def:requests} except that when $N$ recommendations are shown to her, \emph{if} she \emph{does} follow recommendations (i.e., branch $\alpha$), she clicks the item at position $i$ with probability $v_i$.

Note that, in the model considered thus far, it was essentially assumed that $v_i = 1/N$ for all positions $i$.
Incorporating the position preference presents the following complication: before, we simply needed to decide which contents to recommend, captured by control variables $r_{ij}$; now, we need to decide \emph{which content to recommend at which position}, defining $N$ sets of control variables $r_{ij}^n$.


\myitem{Example.} To make the notion of the probabilistic recommendations with positions more concrete, consider a library of $K = 4$ total files. A user just watched item $1$, and $N = 2$ items must be recommended. We now focus on the recommendations of content $1$, so let the first row of the matrix $\mathbf{R}^{1}$ be
$\mathbf{r}_{1}^{1} = [0.0, 1.0, 0.0, 0.0]$ and that of $\mathbf{R}^2$ be $\mathbf{r}_{1}^{2} = [0.0, 0.0, 0.5, 0.5]$. In practice, this means that in position 1 the user will always see content $2$ being recommended (after consuming content $1$), and the recommendation for position 2 will half the time be for content $3$ and half for content $4$.

\myitem{Objective Change.} 
Similarly to Section \ref{sec:problem_formulation}, the transient matrix is now a convex combination of the $N$ recommender matrices as follows $\mathbf{Q} = \alpha \cdot \sum_{n=1}^{N} v_n \cdot \mathbf{R}^n$.
Combining the latter expression and Theorem \ref{thm:total-expected-cost}, the goal is to minimize the expected cost of a long session of such a user 
\begin{equation}\label{eq:cost-rate-2}
\underset{\mathbf{R}^{1},.., \mathbf{R}^{N}}{\textnormal{minimize}}~ \mathbf{p}_0^{T} \cdot \left(\mathbf{I}- \alpha \cdot \sum_{n=1}^N v_n \cdot  \mathbf{R}^n\right)^{-1} \cdot \mathbf{c}
\end{equation}


\myitem{Constraint Changes.} The budget constraint \eq{affine-con} has to change as we have $N$ distinct stochastic matrices whose rows must sum to one (and not to $N$ anymore).

\begin{align}
    \sum_{j=1}^K r_{ij}^{n} = 1, \textnormal{for}~n = 1, \dots, N ~\textnormal{and}~i \in \mathcal{K}.
\end{align}

Regarding the quality constraint Eq. (\ref{eq:quality}), we need to re-define the $q_{i}^{max}$.
To do so, we make use of $\mathcal{U}_{i}(N)$, which is the set of the \emph{$N$ highest $u_{ij}$ values in decreasing order}. Then, the $q_{i}^{max}$ becomes
\begin{align}
q_{i}^{max} = \sum_{n\in\mathcal{U}_{i}(N)} v_n \cdot u_{in}
\end{align}
Thus, a baseline RS (which would achieve $q_{i}^{max}$) would place the most relevant item in the most probable position to be clicked and so on. Additionally, one needs to fix the lhs of the quality constraint Eq. (\ref{eq:quality}), 
and thus the new constraint becomes,
\begin{align}
\sum_{n=1}^N v_n \sum_{j=1}^K r_{ij}^{n} \cdot u_{ij} \ge q \cdot q_{i}^{max}~\forall~i~\in~\mathcal{K},~q \in [0,1]
\end{align}
Furthermore, we need to avoid situations where some content appears simultaneously in more than one positions. As an example, suppose a catalog of $K=4$, and that we need to suggest the user $N=2$ recommendations, and we are interested for the recommendation policy of content $\#1$.
\added{For the sake of argument, say that we had the following policy for item $\#1$: first position in the recommendation list $\mathbf{r}_{1}^{1} = [0.0, 1.0, 0.0, 0.0]$, and the second position $\mathbf{r}_{1}^{2} = [0.0, 0.5, 0.5, 0.0]$. The RS policy for the first position ($\mathbf{r}_{1}^{1}$) dictates to \emph{always} recommend item $\#2$. The second position policy ($\mathbf{r}_{1}^{2}$) dictates to recommend item $\#2$ 50\% of the time (and $\#3$ 50\%), thus leading us to recommend the user item $\#2$ in different slots (50\% of the time more specifically), which is something obviously \emph{unwanted}.}
\added{On the contrary, if the sum of frequencies over the $N$ different slots for item $\#2$ was at most equal to 1, then we would end up with a policy that can always return a set of \emph{different} recommendations. To avoid this situation, we impose $K^2$, additional constraints; each of these constraints upper-bound the sum of recommendation frequencies of every $r_{ij}^{n}$ (along the $N$ slots-positions), so that it is less than, or equal to 1. This is expressed as follows,}
\begin{align}
    \sum_{n=1}^N r_{ij}^{n} \le 1~\forall~i,~j~\in \mathcal{K}.
\end{align}
Finally, as in the uniform click-through probability case, we prohibit the RS from suggesting the same content the user is currently i.e., $r_{ii}^{n} = 0~\forall~i\in\mathcal{K}$ and $n = 1, \dots, N$ which makes up for $K\cdot N$ additional constraints. Wrapping it all up, the optimization problem in hand is the following:
\begin{problem}[OP-pref]
\label{problem:nuct}
\begin{subequations}
\begin{align}
\underset{\mathbf{R}^{1}, \dots, \mathbf{R}^{N}}{\textnormal{minimize}}~~~& {\mathbf{p}_{0}^{T} \cdot (\mathbf{I} - \alpha \cdot \sum_{n=1}^{N} v_n \cdot \mathbf{R}^{n})^{-1}}\cdot \mathbf{c}\label{eq:objective-pp}\\
\textnormal{subject to}~
 & \sum_{j=1}^{K}\sum_{n=1}^{N} v_{n} \cdot r_{ij}^{n} \cdot u_{ij} \geq q\cdot q_{i}^{max},~\forall~i \in \mathcal{K} \label{quality-con1}\\ 
& \sum_{j = 1}^{K} r^{n}_{ij} = 1,~\forall~i \in \mathcal{K}~\textnormal{and}~n=1,...,N \label{affine-con1} \\
& \sum_{n = 1}^{N} r^{n}_{ij} \le 1,~\forall~\{i,j\} \in \mathcal{K} \label{constraint-implementable} \\
& 0 \le r^{n}_{ij} \le 1 \; (i \ne j), \;\; r^{n}_{ii} = 0~\forall~i,~n. \label{box-con1}
\end{align}
\end{subequations}
\end{problem}

\begin{result}\label{res:ltecNu}
\textbf{\nameref{problem:nuct}} is nonconvex: its objective is nonconvex in the variables $\mathbf{R}^1, \dots, \mathbf{R}^N$ and the constraints are linear similarly to \textbf{\nameref{problem:basis}}. Nonetheless, it can also be cast as an LP through the same transformation steps described in Section \ref{sec:journey-optimality} (for details see Appendix).
\end{result}
Note that for relatively small $N$ (around 2, 3, 4) the scale of the problem remains unchanged, as instead of solving a problem with $K^2$ variables, we will now face a problem with $K + N \cdot K^2$ variables. This extra computational burden though, gives us the flexibility to capitalize on the extra knowledge of the $\mathbf{v}$ statistics as we will see later in the simulation section.



\section{Validation Results}
\label{sec:sims}
In this section we observe the how network-friendly RS can actually increase the performance of the cache hit rate (CHR), and more particularly, we focus on the case where the user session is long.

\subsection{The Different Policies}
Throughout the validation section we will consider three different policies.
\begin{itemize}
    \item $\mathbf{P}_1$: The solution of \nameref{problem:myopic} (assumes the user clicks uniformly among the $N$ recommended items).
    \item $\mathbf{P}_2$: Corresponds to the optimal the solution of \nameref{problem:basis} (assumes the user clicks uniformly among the $N$ recommended items).
    \item $\mathbf{P}_3$: The optimal solution of \nameref{problem:nuct} (assumes the user clicks with $v_i$ to the $i$-th position of the recommendations). 
\end{itemize}
We will use the term \emph{Gain} of policy $X$ over policy $Y$ as the following
\begin{align}
Gain = \frac{\text{CHR}_{X} - \text{CHR}_{Y}}{\text{CHR}_{Y}}\times 100 \%
\end{align}

\subsection{Datasets}
We use datasets of video and audio content, to obtain realistic similarity matrices $\mathbf{U}$. We also create some synthetic traces (with similar properties to the real data) for sensitivity analyses.

\myitem{YouTube FR.}
We used the crawler of~\cite{kastanakis-cabaret-mecomm} and collected a dataset from YouTube. We considered 11 of the most popular videos on a given day, and did a breadth-first-search (up to depth 2) on the lists of related videos (max 50 per video) offered by the YouTube API~\cite{ytAPI}. \added{We picked the 11 most popular videos, as this led to trace of $\approx$ 1K contents (i.e., of similar size to our other traces).} We built the matrix $\mathbf{U}\in\{0,1\}$ from the collected video relations by setting $u_{ij}=1$ if the content $j$ is one or two hops away from $i$ through the related list of $i$.

\myitem{last.fm.}
We considered a dataset from the last.fm database \cite{lastfm}. We applied the ``getSimilar'' method to the content IDs' to fill the entries of the matrix $\mathbf{U}$ with similarity scores in [0,1]. We then keep the largest component of the graph. Finally, as the relation matrix we end up is quite sparse, we saturate the values above $0.1$ to $u_{ij} = 1$. This is done in order to have a meaningful $\mathbf{U}$ matrix with many entries. 

\myitem{MovieLens.} We consider the Movielens movies-rating dataset~\cite{movielens}, containing $69162$ ratings (0 to 5 stars) of $671$ users for $9066$ movies. We apply an item-to-item collaborative filtering (using 10 most similar items) to extract the missing user ratings, and then use the cosine distance ($\in[-1,1]$) of each pair of contents based on their common ratings. We set $u_{ij} = 1$ for contents with cosine distance larger than $0.6$.

\myitem{Synthetic.} We genenerate an \emph{Poisson random} graph of content relations $K = 1000$ nodes, where each content/node has on average 8 neighbors.

To accompany our results, we present Table~\ref{table:datasets}: a table with metrics of the content relation graphs we gathered and of the synthetic one we created (just one of size $K=1000$). We define here as $Neighb(i)$ the number of neighbors content $i$ has in the relations graph $\mathbf{U}$.

\begin{table}
\centering
\caption{Dataset Statistics}
\label{table:datasets}
\begin{tabular}{l|c|c|c|c}
			& {Nodes} & {Edges} & {$\overline{Neighb}$} & {$std(Neighb)$}\\
\hline \hline
{MovieLens}			&{1060} &{20162} &{19.02} & {19.61}\\
\hline
{YouTube FR}		&{1059} &{3516} &{3.32} & {9.15}\\
\hline
{last.fm}			&{757} &{5964} &{7.87} & {5.53}\\
\hline
{Synthetic}			&{1000} &{7980} &{7.98} & {2.74}\\
\end{tabular}
\end{table}





\myitem{Simulation Setup.} 
Here, we consider a simple scenario with $c_{i} \in \{0,1\}$, which corresponds to minimizing the cache misses, or equivalently maximizing the CHR. In all of our presented plots, in the $y$-axis we depict the CHR and on the $x$-axis we vary different problem parameters.
Importantly our metric for Section \ref{subsec:part1} will be the CHR of \emph{a long session of requests} as calculated by the objective function Eq. (\ref{eq:objective-infinite-step}) and for Section \ref{subsec:part2} the one calculated from Eq. (\ref{eq:objective-pp}).
Moreover, we assume $\alpha = 0.5 - 0.9$ (\cite{gomez2016netflix}), a Zipf popularity distribution with exponent $s$ (in the range of 0.4-0.8), and that the $C$ most popular contents of the catalog, according to $\mathbf{p}_0$ are cached. 
We highlight that we split the simulations section in two subsections; first one exploring results related to \textbf{\nameref{problem:basis}}, and the second one to \textbf{\nameref{problem:nuct}}.

\subsection{Simulations: Clicking Uniformly over the recommendations}\label{subsec:part1}

In the figures that follow, we vary key parameters of the problem by keeping fixed the remaining ones and see how the CHR metric evolves.


\myitem{Impact of Quality of Recommendations (\textit{q}).}
The most fundamental parameter of the paper is the quality of recommendations a RS provides to its user. To this end, in the first simulation result, see Fig.~\ref{fig:quality}, we increase the quality \% ($x$-axis) constraint and present the CHR performance ($y$-axis) of the two schemes $\mathbf{P}_1$ and $\mathbf{P}_2$ along with the relative gain as described earlier. 
We keep the ratio cache size/catalogue size ($C/K$) and number of recommendations ($N$) fixed throughout. 
Naturally, we observe that less strict quality constraint allows higher flexibility in favoring network-friendly content. 
Hence, Fig.~\ref{fig:quality} shows that for lower values of $q$, the CHR increases both under $\mathbf{P}_2$ and $\mathbf{P}_1$. However, when high-quality recommendations are desired, e.g., $q \ge 70\%$, $\mathbf{P}_2$ heavily outperforms the baseline $\mathbf{P}_1$. This can be easily seen through the curves of \emph{relative gain}, where in all datasets, at $q$ = 95\% we observe a gain of at least 40\%, in Figs.~\ref{fig:sensitivity-quality-synthetic}, \ref{fig:sensitivity-quality-youtubefr}, \added{and more than 50\% in Figs.~\ref{fig:sensitivity-quality-mvlns}, \ref{fig:sensitivity-quality-last}.}

\myitem{Observation 1.}
The impact of $q$ is the most fundamental result of this work. As $q$ grows and the constraint becomes tighter, the margin for cache gain becomes smaller and smaller. That is when employing a policy equipped with look-ahead capabilities shines the most and when a much less sophisticated method fails to lay-over useful content paths through the recommendation mechanism.

\myitem{Observation 2.}
Note here that our relation matrices $\mathbf{U}$ are \emph{binary} in the sense that a content is either related or unrelated. This hints why as $q$ grows, the \emph{Gain} of $\mathbf{P}_2$ over $\mathbf{P}_1$ grows. As $q$ becomes larger, essentially $\mathbf{P}_1$ selects \emph{at random} the related contents it chooses in order to satisfy the constraint whereas $\mathbf{P}_2$ makes its decision based on possible future trajectories of the user.



\myitem{Impact of Number of Recommendations (\textit{N}).}
The YouTube mobile app usually pops 2-3, related videos before the user finishes her current streaming session. For such values of \textit{N} = 2 or 3, in Fig.~\ref{fig:sensitivity-nb-of-recs-uni}, $\mathbf{P}_2$ performs more than 50\% better than the $\mathbf{P}_1$ ($MPH = 20.28\%$), whose performance is not significantly affected by $N$.

\myitem{Observation 3.}
Comparing the two schemes in Fig.~\ref{fig:sensitivity-nb-of-recs-uni} reveals an interesting insight: it is more efficient to nudge the user towards network-friendly content by narrowing down her options $N$ for both network-friendly policies.
Note here that the RS's goal is to find $N$ items that are of high $u_{ij}$ value and are also cached. With the increase of $N$ what happens is the following: suppose that the RS found this \emph{one} item that is useful in all dimensions (cached and related). If $N = 1$, then we can assign the full budget to this particular item and get a cache hit, whereas if $N$ is larger, then due to the randomness with which the user clicks, it is much harder for the RS to drive her towards the neighborhood it wants.

\myitem{Observation 4.}
In the previous observation we briefly explained why the CHR drops for larger $N$ for \emph{any} network-friendly policy. However, it is evident that $\mathbf{P}_2$ is more sensitive in this parameter. This can be explained by the fact that the aforementioned situation of the \say{useful content} is basically just the tip of the iceberg and a very favorable scenario. Essentially what happens most of the time, is that the RS \emph{does not} have such contents, and here is where $\mathbf{P}_2$ does things better: it looks deeper into the session and finds which contents lead to \say{useful} contents in future requests. 

\myitem{Impact of $\alpha$.}
The key message of Fig.~\ref{fig:sensitivity-alpha-lastfm} ($MPH = 5\%$) is that a \emph{multi-step vision} method such as $\mathbf{P}_2$ takes into account the knowledge of user's behavior ($\alpha$). This can be mainly seen by the superlinear and linear improvement of $\mathbf{P}_2$ and $\mathbf{P}_1$ methods respectively.

\myitem{Impact of zipf parameter.}
\added{In the previous plots, we have assumed a zipf parameter from 0.5 to 0.7. In the result we show next, we fix all parameters and increase the popularity vector skewdness. As a result, the hit rate of both $\mathbf{P}_1$ and $\mathbf{P}_2$ will increase, as the caching is based solely on popularity. Interestingly, the hit rate increases superlinearly, but similarly to earlier, our focus is on the performance gain when the RS policy has look-ahead capabilities. The result of Fig.~\ref{fig:sensitivity-zipf-mvlns} for the Movielens dataset, validates that the proposed policy outperforms the myopic one in the entire range of the simulated values.}

\begin{figure}
\centering
\subfigure[Param.: $N=2,~s=0.7,~C/K=1.5\%,~\alpha=0.8$,~$MPH$=6.5\%]{\includegraphics[width=0.4\columnwidth]{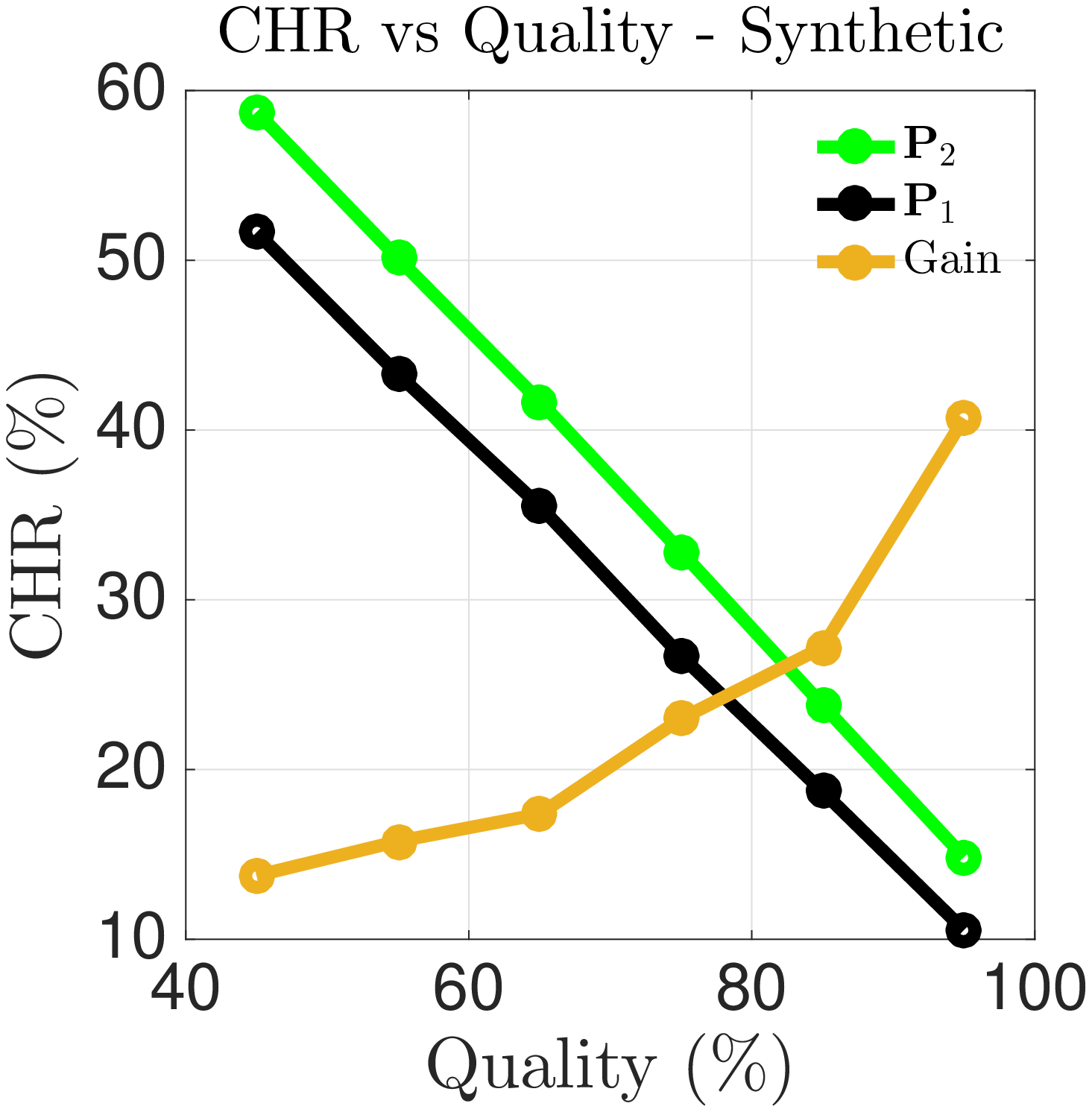}\label{fig:sensitivity-quality-synthetic}}
\hspace{0.08\columnwidth}
\subfigure[Param.: $N=3,~s=0.5,~C/K=0.93\%,~\alpha=0.8$~$MPH$=5.3\%]{\includegraphics[width=0.4\columnwidth]{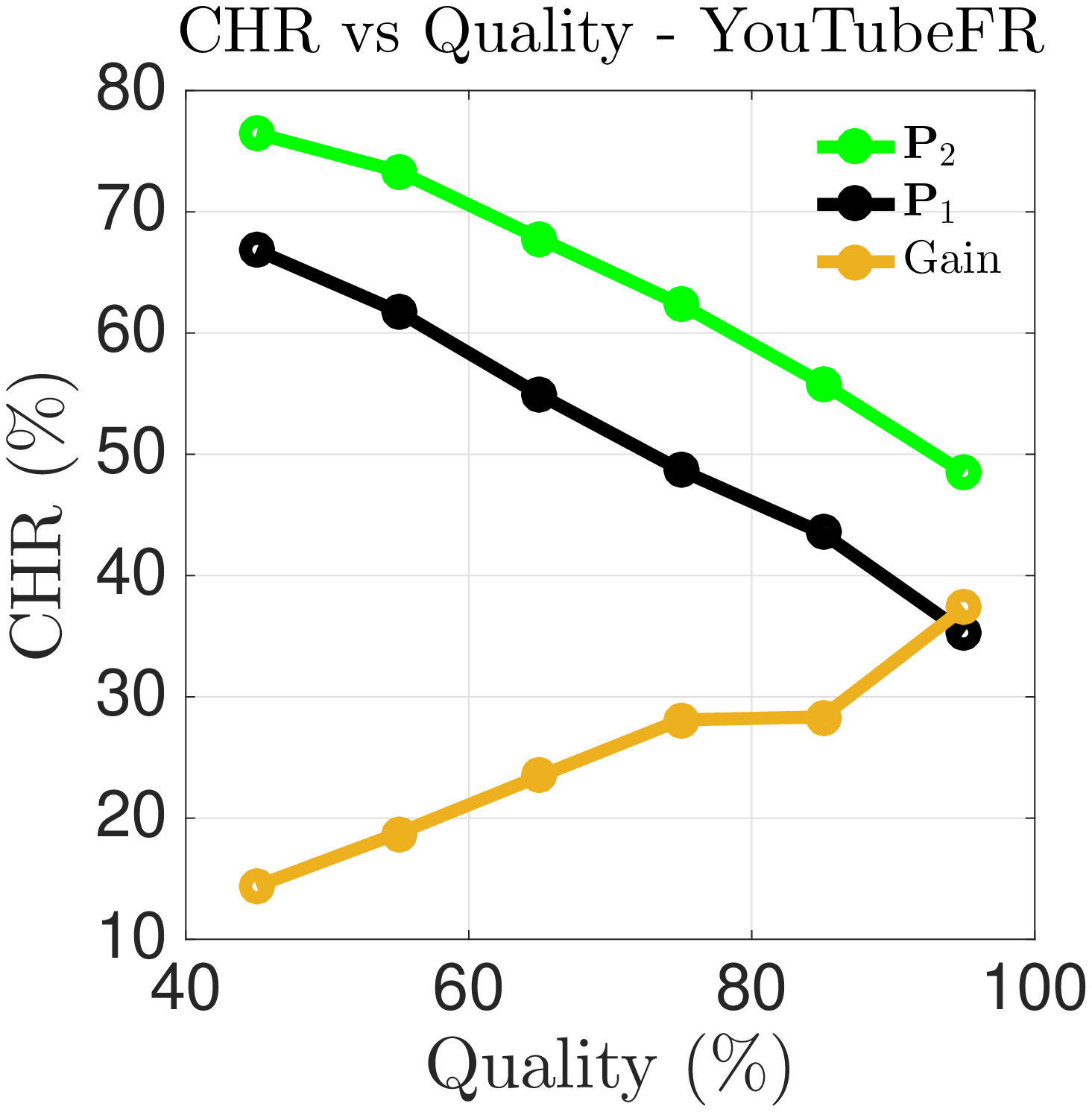}\label{fig:sensitivity-quality-youtubefr}}
\hspace{0.08\columnwidth}
\subfigure[\added{Param.: $N=2,~s=0.6,~C/K=1\%,~\alpha=0.7$,~$MPH$=11.5\%}]{\includegraphics[width=0.4\columnwidth]{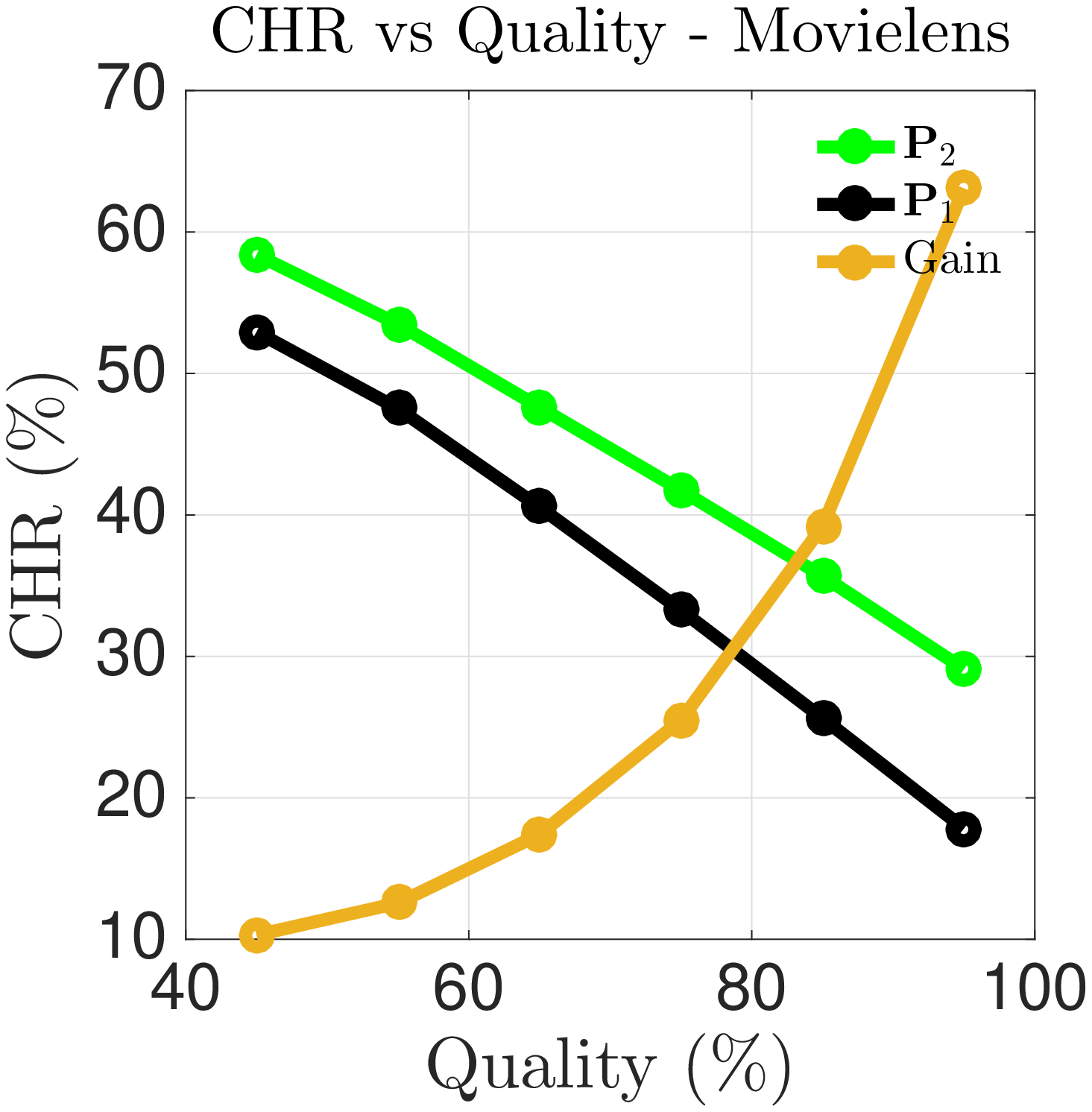}\label{fig:sensitivity-quality-mvlns}}
\hspace{0.08\columnwidth}
\subfigure[\added{Param.: $N=3,~s=0.6,~C/K=1\%,~\alpha=0.7$,~$MPH$=10.8\%}]{\includegraphics[width=0.4\columnwidth]{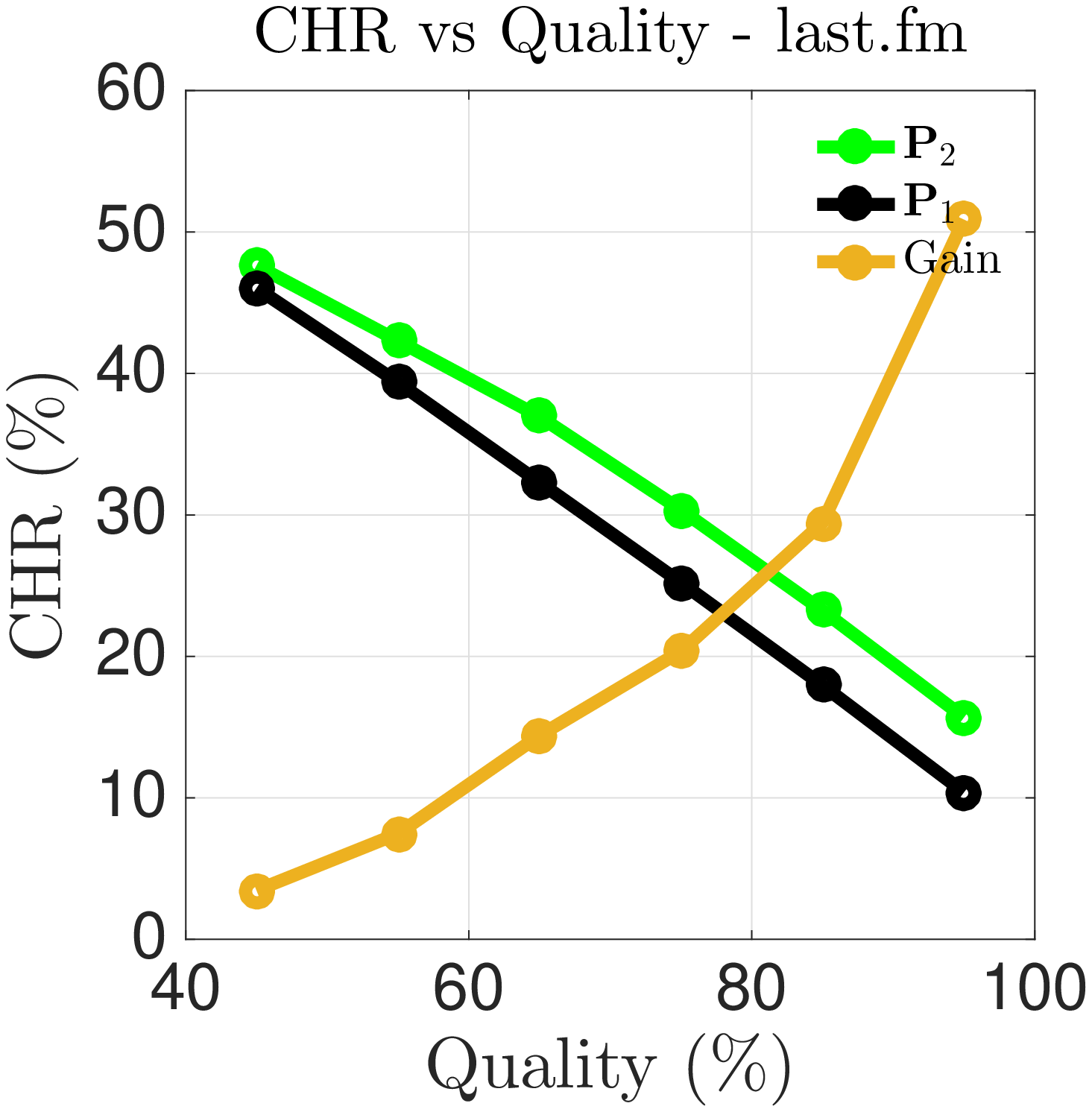}\label{fig:sensitivity-quality-last}}
\vspace{-2.5mm}
\caption{Cache Hit Rate vs Quality, (a): $K=1000$,~Synthetic, (b): $K=1100$,~YouTubeFR, (c): $K=1060$,~Movielens, (d): $K=757$,~last.fm}\label{fig:quality}
\end{figure}

\begin{figure*}
\centering
\subfigure[Param.:  $q=85\%,~s=0.7,~C/K=1.5\%,~\alpha=0.8$]{\includegraphics[width=0.4\columnwidth]{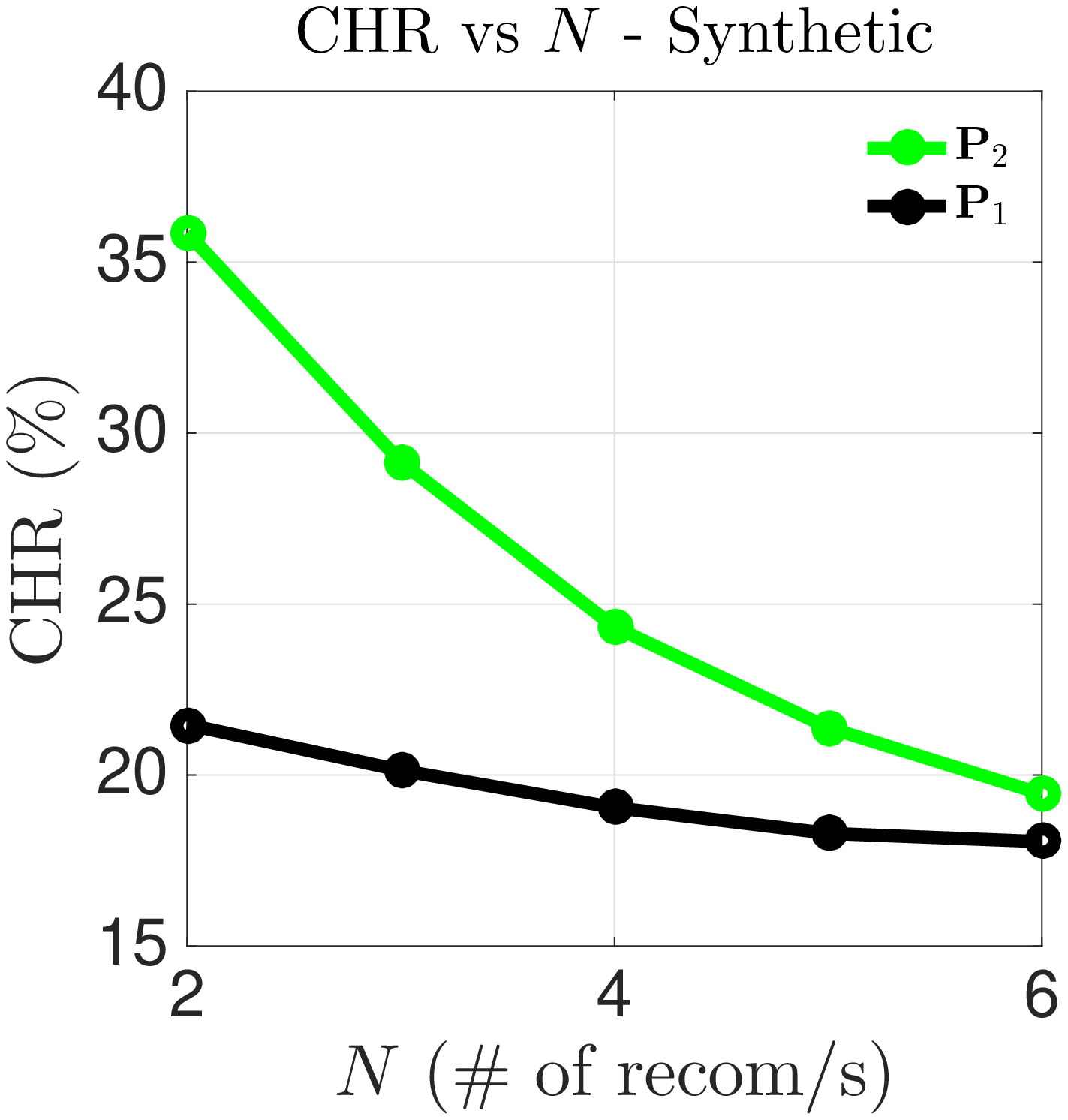}\label{fig:sensitivity-nb-of-recs-uni}}
\hspace{0.08\columnwidth}
\subfigure[Param.: $N=2$,~$C/K=1\%$,~$s=0.6$,~$q=90\%$]{\includegraphics[width=0.4\columnwidth]{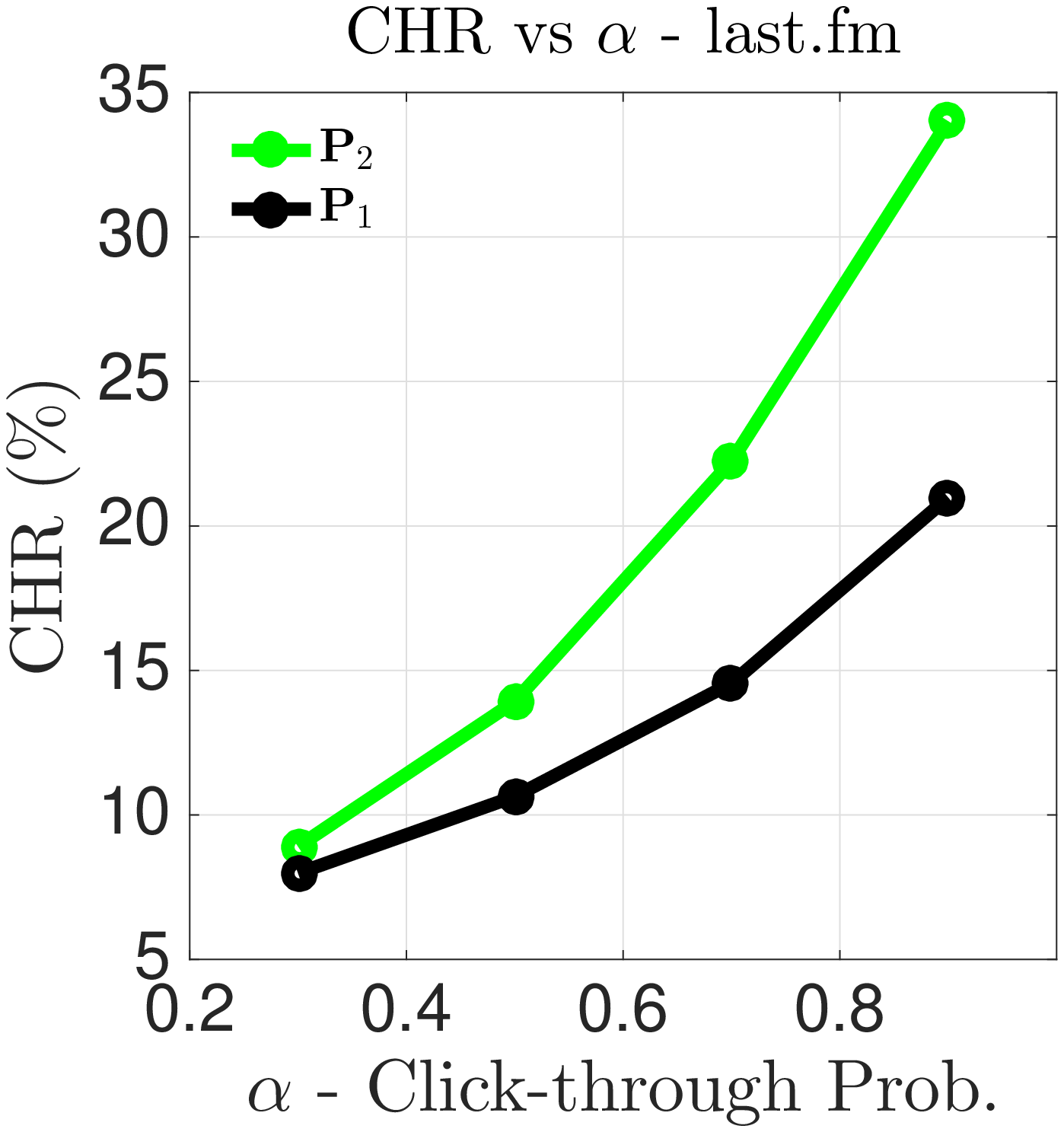}\label{fig:sensitivity-alpha-lastfm}}
\hspace{0.08\columnwidth}
\subfigure[\added{Param.: $N=3$,~$\alpha=0.7$,~$q=85\%$, $C/K=1\%$}]{\includegraphics[width=0.4\columnwidth]{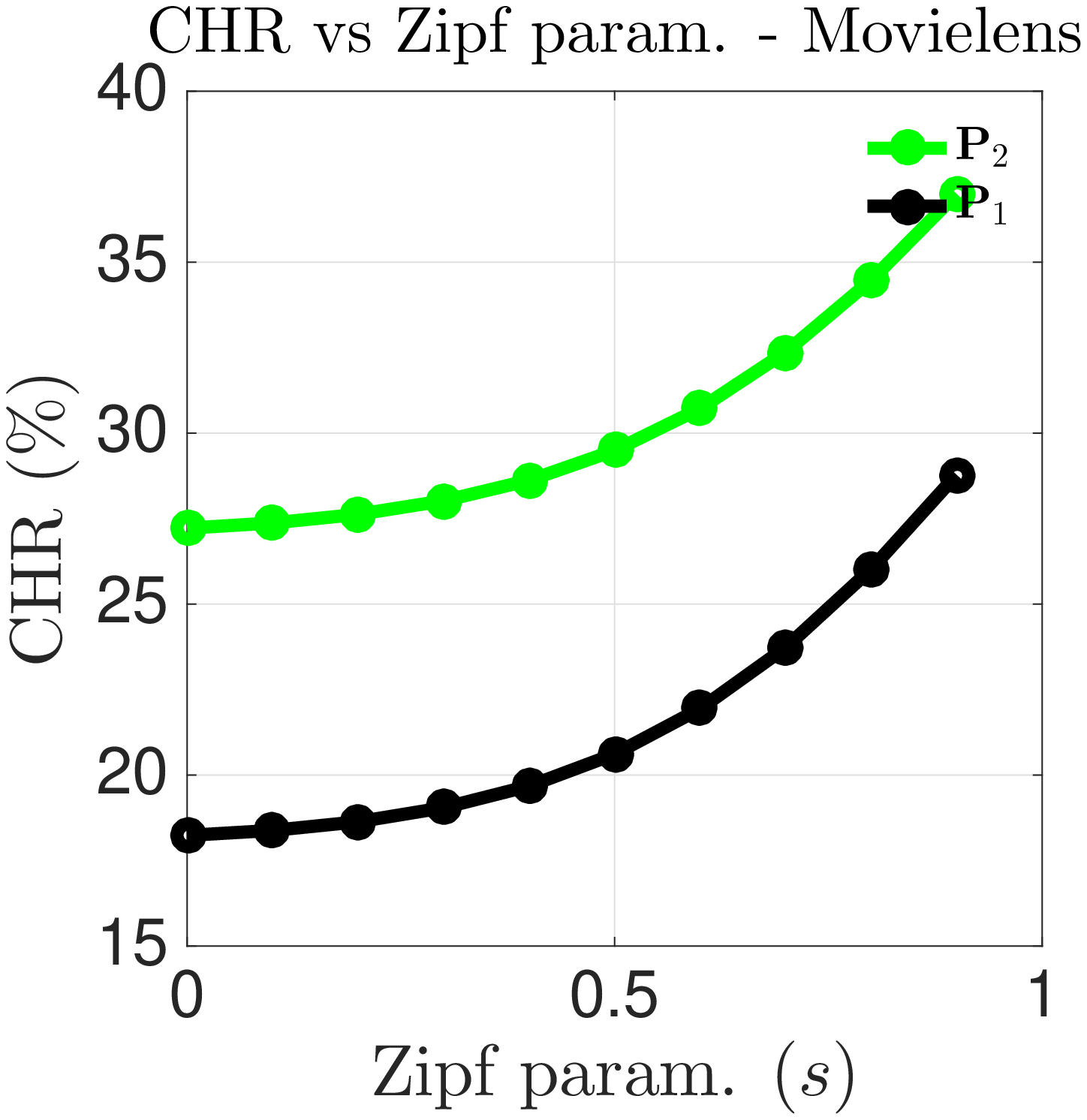}\label{fig:sensitivity-zipf-mvlns}}
\caption{Cache hit Rate vs (a): $N$ (Synthetic), (b): $\alpha$ (last.fm), (c): zipf param. (Movielens)}
\end{figure*}




\subsection{Simulations: Clicking Non-Uniformly over the recommendations}\label{subsec:part2}
In Section \ref{sec:optimization_methodology}, we establish theoretically why our method has clear benefits in the regime of \emph{long user session} by showing optimality of $\mathbf{P}_2$. In the previous subsection, we presented some results to show \emph{in terms of actual numbers}, how much of an improvement a method with deep vision such as $\mathbf{P}_2$ can have, over a short-sighted one such as $\mathbf{P}_1$.

In this subsection, we will slightly change direction and try to understand whether the knowledge of $\mathbf{v}$ can deliver even further gains in the regime of \emph{long sessions}. To this end, we will investigate policies $\mathbf{P}_3$ and $\mathbf{P}_2$, which both have look-ahead capabilities, but the first is aware of $\mathbf{v}$ while the second assumes uniform click over the $N$ recommended items. Moreover, we will employ as a parameter, the entropy of the pmf $\mathbf{v}$ which is defined as

\begin{align}
    H_{\mathbf{v}} = H(v_1,..,v_N) = -\sum_{n=1}^{N} v_n \cdot \log(v_n). 
\end{align}

According to $H_{\mathbf{v}}$, a user who clicks uniformly between any of the $N$ positions has the maximum entropy $H_{\mathbf{v}} = 1$, whereas a user who clicks only in one position (e.g., the first up on the screen) has the minimum entropy $H_{\mathbf{v}} = 0$ as she clicks deterministically.

In Figs.~\ref{fig:opt-agno-mvlns},~\ref{fig:relative-vs-cars} (see Table~\ref{table:parameters} for simulation parameters), we assume behaviors of increasing entropy; starting from users that show preference on the higher positions of the list (low entropy), to users that select uniformly recommendations (maximum entropy). In our simulations, we have used a zipf distribution~\cite{RecImpact-IMC10} over the $N$ positions and by decreasing its exponent, the entropy on the $x$-axis is increased. As an example, in Fig.~\ref{fig:opt-agno-mvlns}, lowest $H_\mathbf{v}$ corresponds to a vector of probabilities $\mathbf{v} = [0.8,0.2]$ (recall that $N = 2$), while the highest one on the same plot to $\mathbf{v} = [0.58,0.42]$.

Thus, we initially focus on answering the following basic question: \emph{Is the non-uniformity of users' preferences to some positions helpful or harmful for a network friendly RS?} From Figs.~\ref{fig:opt-agno-mvlns},~\ref{fig:relative-vs-cars}, it becomes very clear that the lower the entropy, the more the $\mathbf{v}$-awareness helps $\mathbf{P}_3$ gain over the agnostic policy $\mathbf{P}_2$.

\begin{table}
\centering
\caption{Parameters: Figure \ref{fig:no1}.} \label{table:parameters}
\begin{tabular}{l|l|l|l|l|l}
			& {$q \%$} & {$zipf(s)$} & {$\alpha$} & {$N$} & {MPH \%} \\
\hline \hline
{MovieLens}			&{80} &{0.8} &{0.7} & {2} & {23.26} \\
\hline
{YouTube FR}		&{95} &{0.6} &{0.8} & {2} & {12.17}\\
\hline
{last.fm}			&{80} &{0.6} &{0.7} & {3} & {11.74}\\
\end{tabular}
\end{table}

\myitem{Observation 1.} We observe by these plots that a skewed $\mathbf{v}$, is helpful for the NFRS. In the extreme case where  $\mathbf{v}$ is extremely skewed ($H_{\mathbf{v}} \to 0$), where virtually this means $N = 1$, the user clicks deterministically, 
and the optimal hit rate becomes maximum. 
This can be also validated in
Fig.~\ref{fig:sensitivity-nb-of-recs}, 
where for increasing entropy the 
the hit rate decreases and its maximum is attained for $N=1$.

\begin{figure} 
\centering
\subfigure[Absolute Perf.]{\includegraphics[width=0.4\columnwidth]{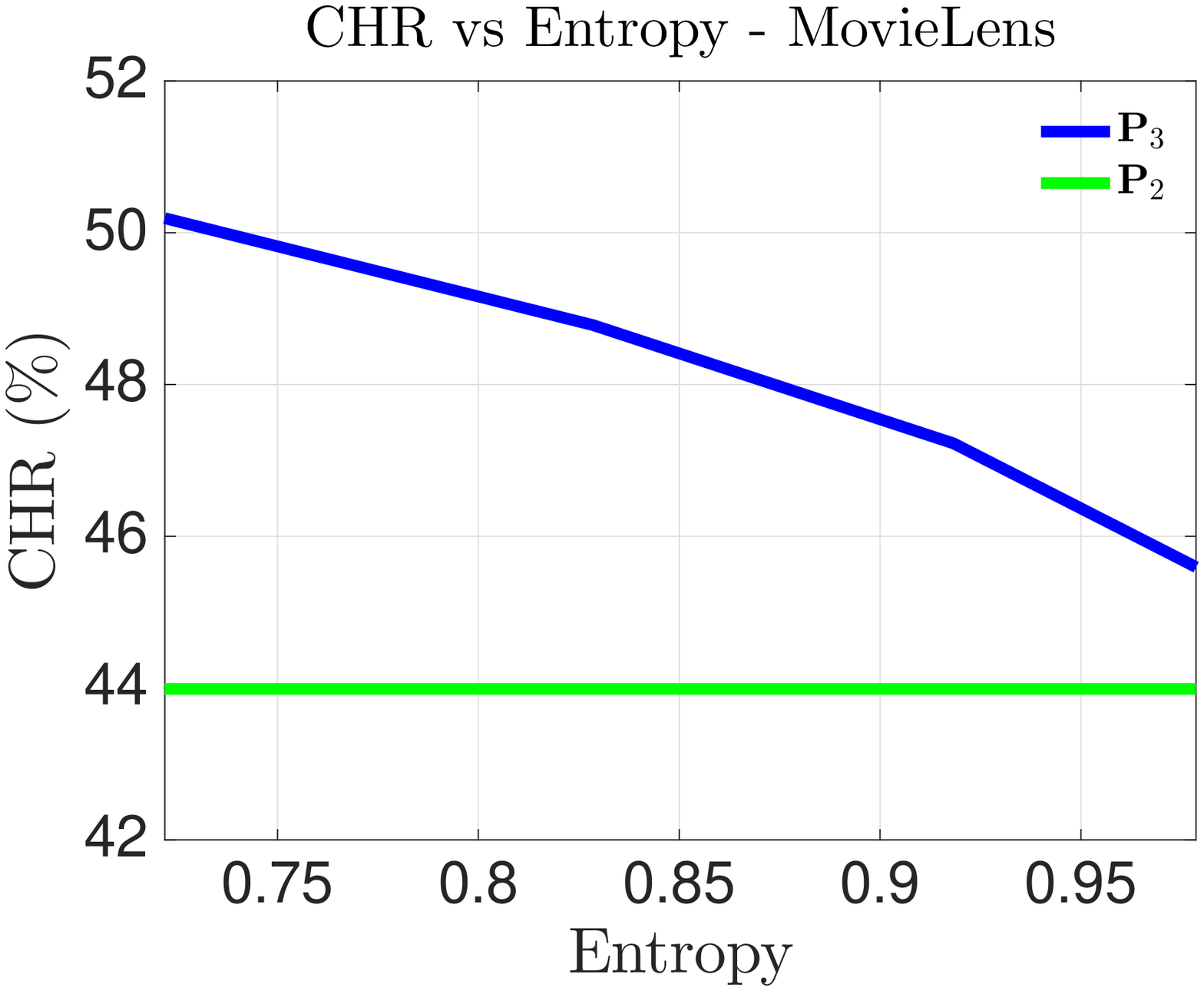}\label{fig:opt-agno-mvlns}} 
\hspace{0.08\columnwidth}
\subfigure[Relative Gain \%]{\includegraphics[width=0.4\columnwidth]{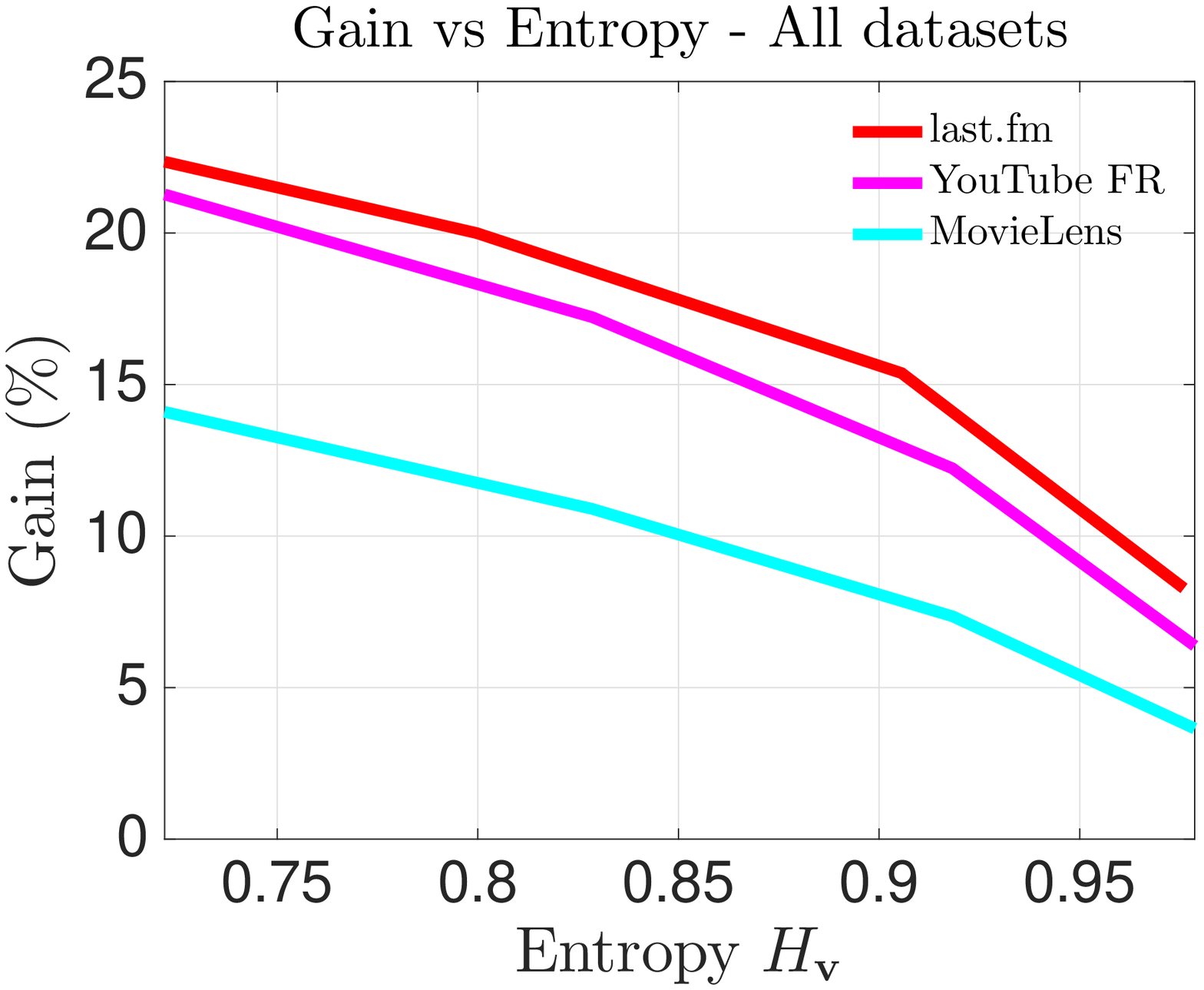}\label{fig:relative-vs-cars}}
\vspace{-2.5mm}
\caption{Cache Hit Rate vs~ $H_{\mathbf{v}}~(C/K\approx1.00\%)$}
\label{fig:no1}
\end{figure}

Lastly, we investigate the sensitivity of $\mathbf{P}_3$ and $\mathbf{P}_2$, against the number of recommendations ($N$). 
In Fig.~\ref{fig:sensitivity-nb-of-recs}, we present the CHR curves of the two schemes for increasing $N$, where we keep constant the distribution $\mathbf{v} \sim zipf(0.9)$. As expected, for $N = 1$ (e.g., YouTube autoplay scenario) $\mathbf{P}_2$ and the proposed scheme coincide, as there is no flexibility in having \emph{only one} recommendation. 


\myitem{Observation 2.} For large $N$, $\mathbf{P}_2$ may offer the \say{correct} recommendations (cached or related or both), but it cannot place them in the right positions, as there are now too many available spots. 
In contrast, the scheme $\mathbf{P}_3$ recommends the \say{correct} contents, and places the recommendations in the \say{correct} positions. 
Fig.~\ref{fig:heatmap}, strengthens even more the Observation 2; its key conclusion is that with high enough enough $s$ (i.e. low $H_{\mathbf{v}}$) and more than 2 or 3 recommendations, while $\mathbf{P}_2$ aims to solve the multiple access problem, its \emph{position preference unawareness} leads to highly suboptimal recommendation placement, and thus a severe drop of its CHR performance compared to the $\mathbf{P}_3$.

\begin{figure} \label{fig:combination-plots}
\centering
\subfigure[$q=80\%,~K=400$]{\includegraphics[width=0.4\columnwidth]{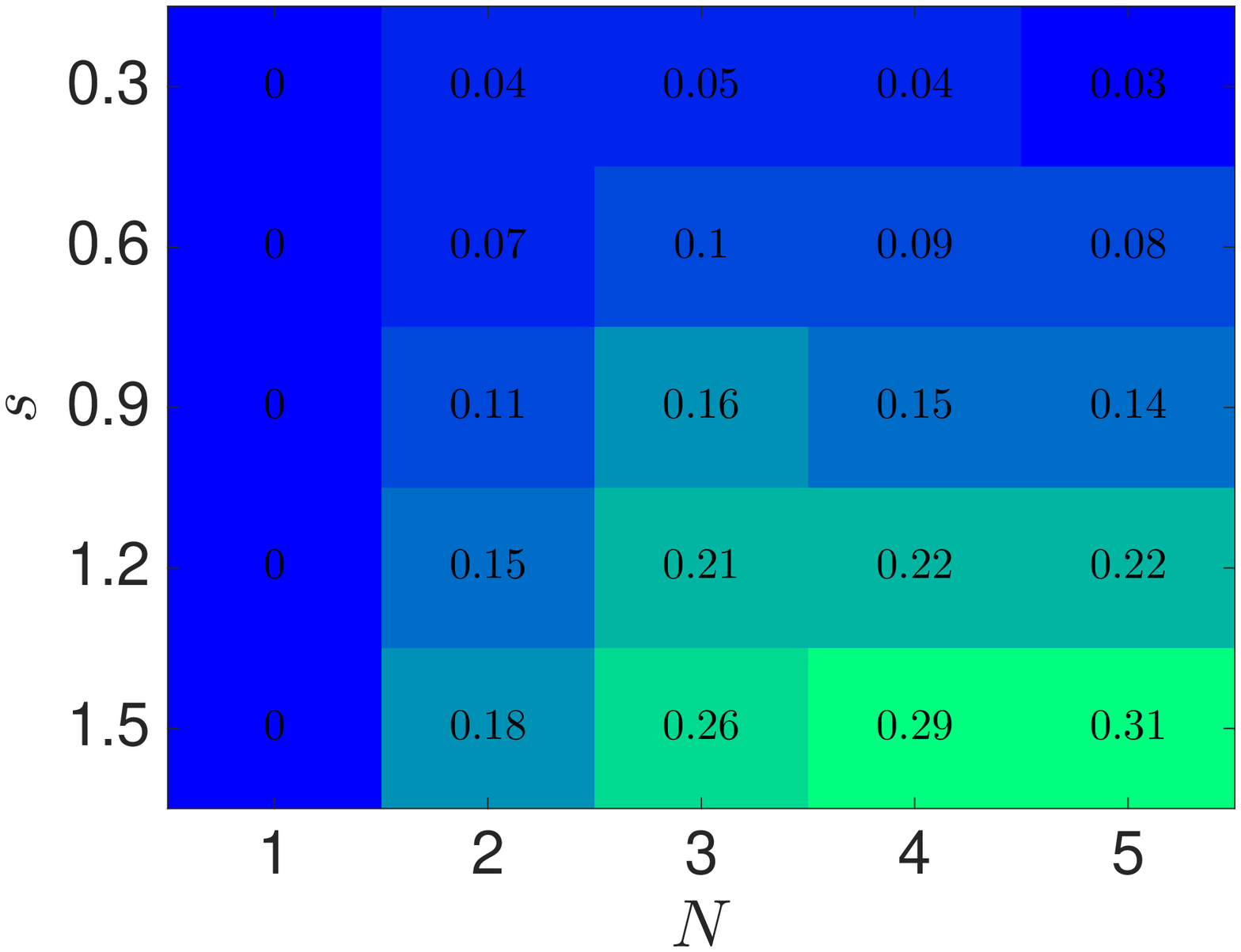}\label{fig:heatmap}}
\hspace{0.08\columnwidth}
\subfigure[Absolute Perf. ($q=90\%,~s = 0.6,~MPH=11.24\%$)]{\includegraphics[width=0.4\columnwidth]{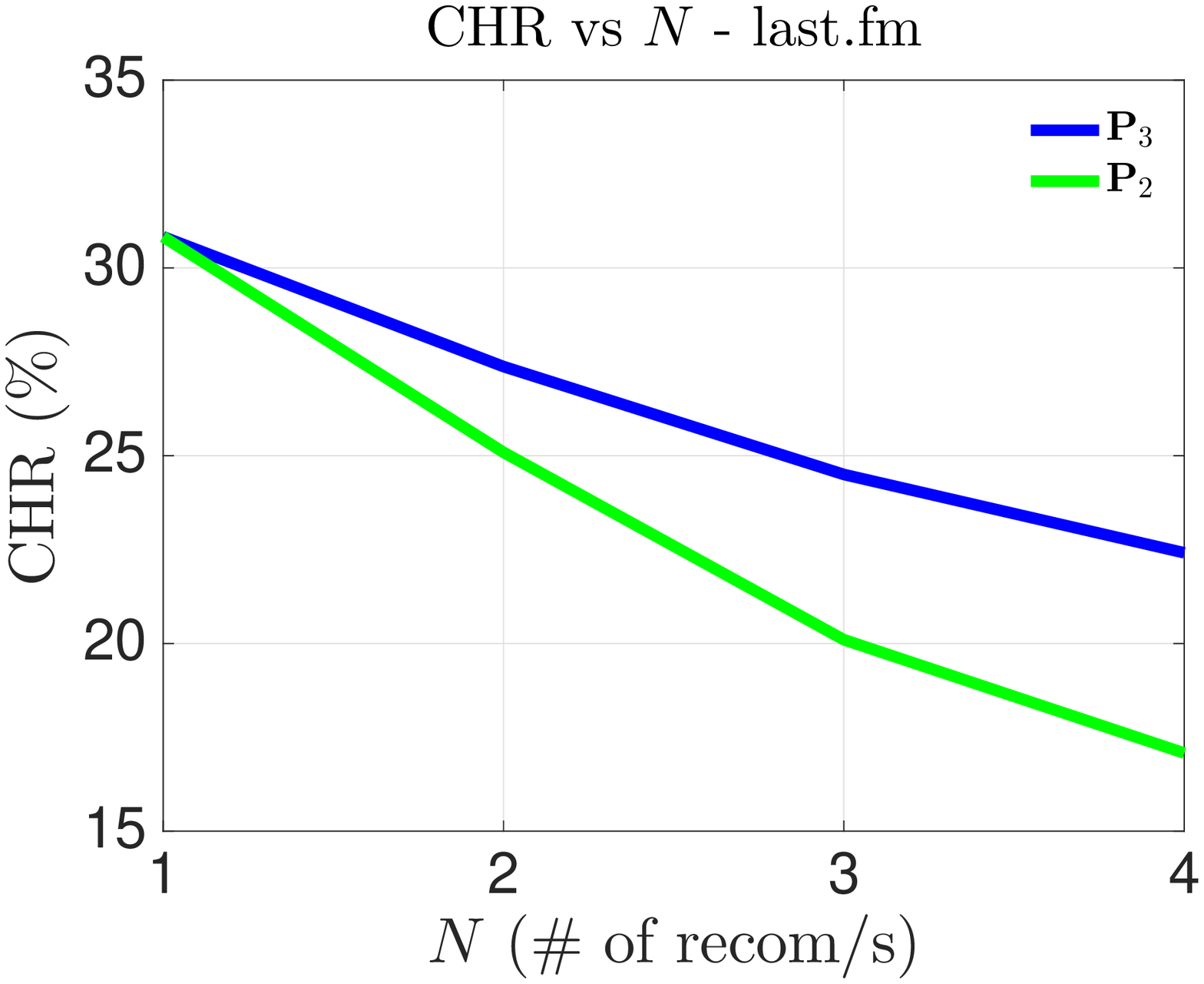}\label{fig:sensitivity-nb-of-recs}}
\vspace{-2.5mm}
\caption{(a:) Relative Gain vs $(N,\beta)$ and (b:) Cache Hit Rate vs~$N$ ~($C/K\approx 1.00\%,~\alpha=0.7$)}
\end{figure}

\section{Related Work}
\label{sec:related}
\myitem{Recommendation and Caching Interplay.} 
The relation between recommendation systems and caching has only recently been considered~\cite{chatzieleftheriou2019TMC, cache-centric-video-recommendation,munaro2015content,sermpezis2018soft,guo2017caching,liu2018learning,kastanakis-cabaret-mecomm,song2019interactions}. 
\added{The problem of optimizing \emph{jointly} caching and recommendations in a static (IID requests and static placement) setting, for a user accessing the home screen of an application has recently been considered~\cite{chatzieleftheriou2019TMC}, and a heuristic approach was suggested for its solution.}
\added{In~\cite{cache-centric-video-recommendation}, the authors propose a simple yet effective algorithm according to which, they randomly inject cached content in the \say{related list} of YouTube, which comes with a small decrease of recommendations quality, and an interesting increase in cache hit rate.}
\added{In~\cite{sermpezis2018soft}, the focus is on content caching, and the authors introduce the notion of \say{soft cache hit}. Essentially, the aim is to cache the items that are popular, but also hot in the sense of \emph{frequently appearing in the related list of other contents}.} 
%
\added{A more measurement-oriented approach is found in~\cite{kastanakis-cabaret-mecomm}, where a true cache-aware RS was implemented, and based on the authors findings, high gains can be achieved when the RS suggests content from the cache.}
In~\cite{liu2018learning}, a single access user is considered. There, caching policy is based on machine learning techniques, where the users' \textit{behavior} is estimated through the users' interaction with the recommendations and this knowledge is then exploited at the next edge cache updates. 
%
%
In \cite{song2019interactions}, the authors introduce a first of its kind formulation of a wireless recommendations problem in a contextual bandit framework, which they call contextual broadcast bandit. In doing so, they propose an epoch-based algorithm for its solution and show the regret bound of their algorithm. Interestingly, they conclude that the user preferences/behaviors learning speed is proportional to the square of available bandwidth.
The work in \cite{chatzieleftheriou2019TMC} considers the \emph{joint} problem of caching and recommendations in a static setting. Similarly in~\cite{liu2018learning}, a single access user is considered. There, caching policy is based on machine learning techniques, where the users' \textit{behavior} is estimated through the users' interaction with the recommendations and this knowledge is then exploited at the next edge cache updates. 

\myitem{Recommendations for the Long Term Cost.}
In~\cite{giannakas2018show}, the problem of RS design for minimum network cost in the long run first appears. We formulated a \emph{nonconvex} problem, and proposed a heuristic ADMM algorithm on the nonconvex formulation which comes \emph{with no theoretical guarantees}.
\added{In~\cite{giannakas2019wiopt}, the preliminary version of this work, we convexified the same problem and offered an \emph{optimal} solution under an LP formulation. Here, through means of simulation, we show that the LP framework is \emph{computationally more efficient} than the ADMM previously used in~\cite{giannakas2018show}, as State of Art LP solvers can be used to tackle it.}

Equally importantly, we extended the framework in order to capture the importance of the recommendations placement in the GUI, again \emph{optimally}~\cite{giannakas2019wiopt}. Note that several of the aforementioned studies~\cite{giannakas2018show,sermpezis2018soft,munaro2015content} ignore the preference that users exemplify to some the position of recommendations over others. The work in~\cite{chatzieleftheriou2019TMC}, while taking into account the ranking of the recommendations in the modeling and their proposed algorithm, in the simulation section they assume that the boosting of the items is equal. 

\added{Overall, this work serves as a unification, as it extends and generalizes our previous works by including new results from additional datasets, that further strengthen the case of LP-based NFR in long user sessions.}




\myitem{Optimization Methodology.} 
The problem of optimal recommendations for multi-content sessions, bares some similarity with PageRank manipulation~\cite{ermon2014designing,fercoq2013ergodic,avrachenkov2006effect}. The idea there is to choose the links of a subset of pages (the user has access to) with the intention to increase the PageRank of some targeted web page(s). Although that problem is generally hard, some versions of the problem can also be convexified~\cite{fercoq2013ergodic}. 


\section{Discussion}
\label{sec:discussion}
\added{In this section we discuss some open problems that are related to NF-RS design. It is important to highlight that most of the following problems are essentially generalizations of what we have presented in this work, and in general not addressable by our framework.} 

\myitem{Joint optimization of caching and recommendation: Long user sessions.}
A reduction of the general network-friendly recommendations is to consider it in the context of caching, i.e., the network cost becomes the cache miss probability. In our framework, it is not specified whether we focus on the edge caching problem (single or femto) or a CDN-like architecture. We simply need from the operator, who is aware of the \emph{network state}, the set of content costs and we guarantee to deal with the recommendation side of things. However, a problem that is timely problem is the one of joint cache placement and recommendation under the regime of many content requests. The problem we studied in this paper, could be cast in the single cache framework where the $\mathbf{c}$ is the hot vector deciding which contents should be cached and which not. It is quite evident though that this problem is a mixed integer program ($\mathbf{c}$ is a binary vector), where the objective is quadratic (caching decisions are multiplied with the recommendations matrix), and therefore lies in the category of \emph{hard} problems. Nonetheless, heuristic methods that alternatingly optimize over the two variables still apply; \added{one could even use the method proposed in this paper for the minimization of the recommendation variables}. On the other hand, the joint \emph{femtocache} content placement and recommendation would require additional modeling work since our optimization objective (if seen under the lens of caching) does not consider different base stations.

\myitem{Different user models.}
In this study, we focused on a simple, yet quite general user behavior. The user model we discussed is fully parameterized and essentially if one has measured data/statistics about the crucial quantities like $\alpha$ (user willingness to click on recommended items) or under which $q$ is the user happy, our modeling and optimization approach enjoys optimality guarantees. However, recommendations and how they affect the user's request pattern is an open problem. A major assumption we made here, was that users have fixed $\alpha$, which could be considered as unrealistic. In practice, users have a more reactive behavior towards their recommendations; receiving good recommendations might increase their instantaneous $\alpha$ or bad ones might have the opposite effect. Thus, starting off from the Markov Chain framework, and based on the modeler belief for the user, one could engineer different models. A first nontrivial extension to our model would be to define a user whose $\alpha$ (i.e., clickthrough on recommendations) is policy dependent; thus instead of modeling the $\alpha$ as constant and incorporating the quality of recommendations as an external hard constraint, we could embed it in the content transition and make $\alpha$ a function of the policy. Moreover, here we have discussed the cases where the item selection is either random or depends in an i.i.d manner from the position the item is placed. An additional modeling twist would be to allow the item selection to be based on the similarities $u_{ij}$ of the items. However, these extensions would further complicate things as they would introduce additional nonconvexities to our \nameref{problem:basis}.


\myitem{Dynamic (caching) conditions.}
According to our assumptions, the network state, i.e., the costs of the contents remain the same throughout the course of the day. However, in many practical scenarios the operators already have their infrastructure inside the network where some dynamic caching policy such as LRU, $q$-LRU or LRU-$K$~\cite{che2002hierarchical, o1993lru} is pre-implemented. A big challenge that remains unresolved is the one of designing optimal recommendation policies over a network of LRU-like cache replacement policies.
In our opinion, this framework has two open questions. First of all, as dynamic cache policies are not easy to analyze, approximations are typically employed in order to acquire meaningful metrics. In the dynamic cache setting, the well known \emph{Che} and \emph{time-to-live} approximations do not capture the effect of a RS over the average lifetime of contents inside the cache. So we consider analyzing the effect of \emph{any} RS on the cache lifetime statistics to be an interesting topic on its own. Furthermore, as we know the RS has the power to \emph{shape} the content popularity and therefore which contents \emph{will live longer} inside the cache. Nonetheless, now the stage for the recommendation algorithm is much more hostile. Imagine we kept our Markov chain framework and \emph{augment its state space to include also the cache configuration}. Then plausibly, one would want to find the static (computed offline) optimal recommendation policy under an LRU caching policy. However, the cache configurations is the exploding in size unique permutations of $\perm{K}{C}$, resulting to a state space of size $K \cdot \perm{K}{C}$. It becomes obvious that even for a moderate problem size such as $K = 200$ and $C = 3$ we have $\approx 1.5$ billion states, and over \emph{each one of them} we should make decisions. Thus, either the content lifetime approximation should be somehow used as a proxy to the content cost or maybe even some function approximation in order to discover what features of the problem \emph{really matter}.


\myitem{Unknown, static or dynamic, user behavior and Learning.}
The solution mindset we employed is the one of model-based optimization. As such, our solution might need re-tuning with the change of $\alpha$ during the course of the day. Problems like this, could be better handled through learning-based (a.k.a. model-free) optimization methods. Although this path sounds very appealing, there are a few pitfalls into it. As an example, if we employ a Q-Learning based algorithm, a question that arises is \say{Will it converge soon enough?} Maybe by the time it has learned the user behavior, the network state has changed and then all the effort on learning the user might have gone to waste. \added{Thus, such approaches cannot be applied straightforwardly in our problem and could be considered as new problems on their own. A promising idea for such a demanding problem would be learn policies through function approximation, which can generalize~\cite{sutton1998introduction}. This line of problems could also be faced through Online Convex Optimization methods, which are well known to optimize some dynamically changing function, even if it is picked by an adversary~\cite{hazan2019introduction}.}




\subsection*{Acknowledgments}
This research is funded by the ANR \say{5C-for-5G} project under grant ANR-17-CE25-0001, and the IMT F\&R, \say{Joint Optimization of Mobile Content Caching and Recommendation} project. It is also co-financed by Greece and the European Union (European Social Fund- ESF) through the Operational Programme ``Human Resources Development, Education and Lifelong Learning'' in the context of the project ``Reinforcement of Postdoctoral Researchers - 2nd Cycle'' (MIS-5033021), implemented by the State Scholarships Foundation (IKY).

\bibliographystyle{ieeetr}
\bibliography{nfr}
\vskip 0pt plus -1fil
\begin{IEEEbiography}[{\includegraphics[width=1in,height=1.25in,clip,keepaspectratio]{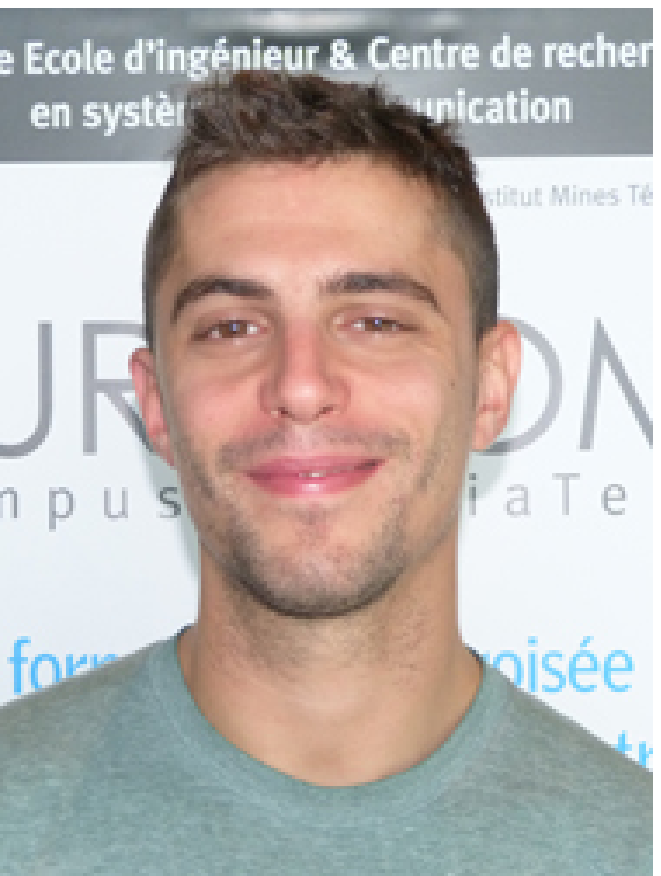}}]{{Theodoros Giannakas}}
received the Diploma in Electrical and Computer Engineering from the University of Patras, Greece, his MSc in Wireless Communications from the University of Southampton, UK and his PhD in Computer Science and Networks from EURECOM, Sophia Antipolis, France, where he is currently working as a post-doctoral researcher. His main research interests includes modeling and optimization, for network friendly recommendation systems and network slicing.
\end{IEEEbiography}
%
%
\begin{IEEEbiography}[{\includegraphics[width=1in,height=1.25in,clip,keepaspectratio]{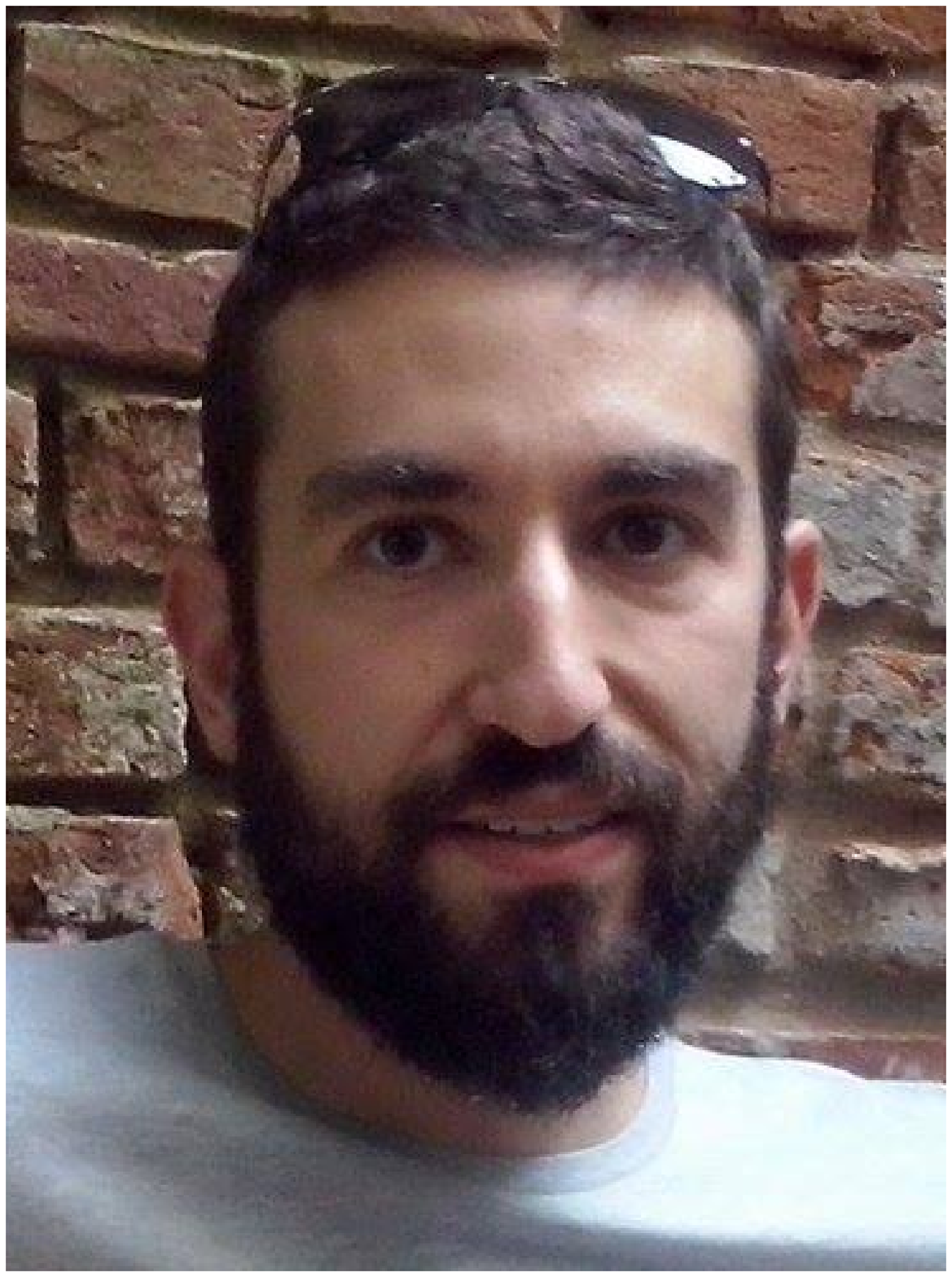}}]{Pavlos Sermpezis} received the Diploma in Electrical and Computer Engineering from the Aristotle University of Thessaloniki, Greece, and a PhD in Computer Science and Networks from EURECOM, Sophia Antipolis, France. He is currently a post-doctoral researcher at Datalab, Department of Informatics, Aristotle University of Thessaloniki, Greece. His main research interests are in modeling and performance analysis for communication networks, and data science.
\end{IEEEbiography}
%
%
\begin{IEEEbiography}[{\includegraphics[width=1in,height=1.25in,clip]{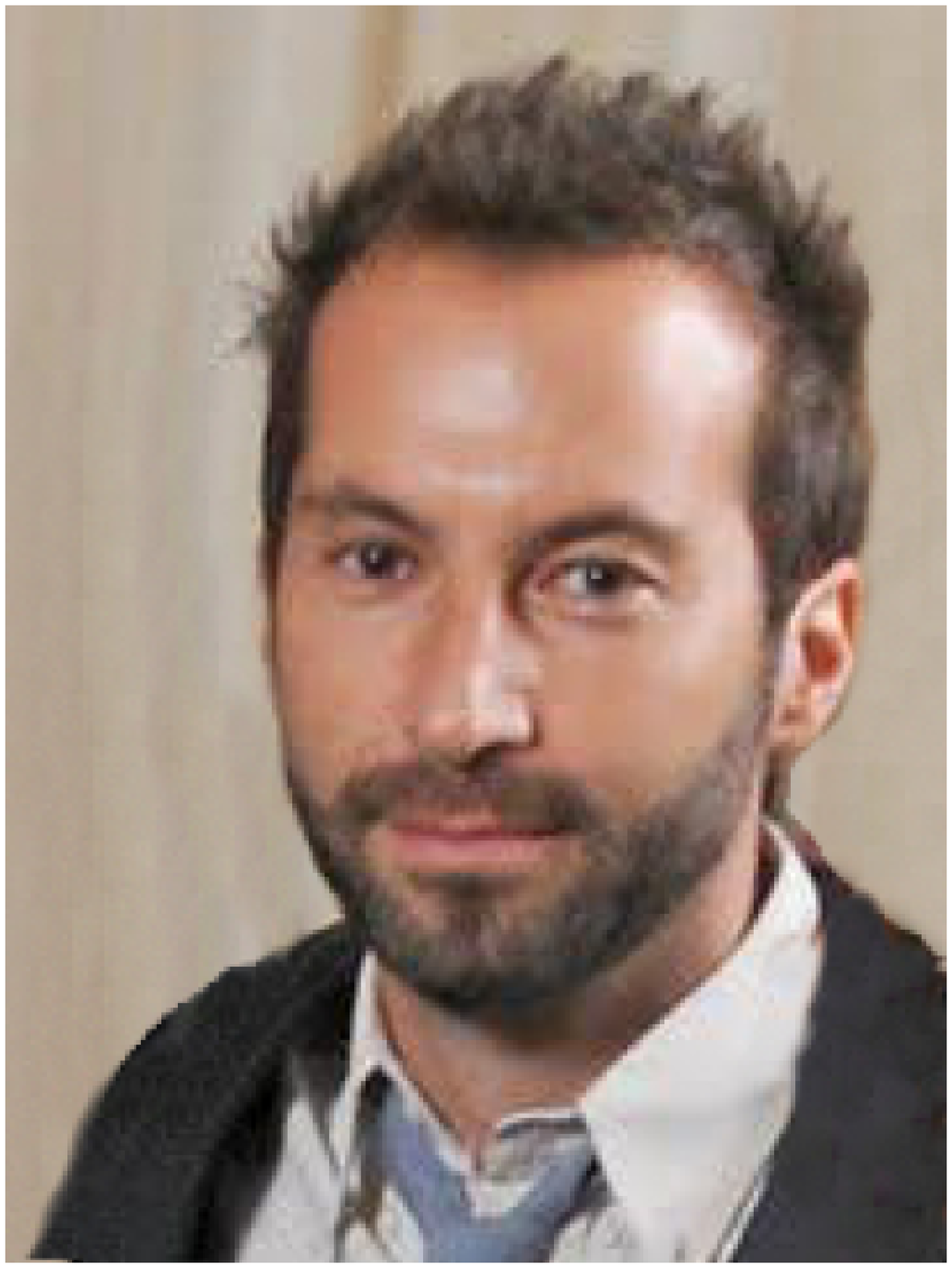}}]{Thrasyvoulos Spyropoulos}
received the Diploma in Electrical and Computer Engineering from the National Technical University of Athens, Greece, and a Ph.D degree in Electrical Engineering from the University of Southern California. He was a post-doctoral researcher at INRIA and then, a senior researcher with the Swiss Federal Institute of Technology (ETH) Zurich. He is currently an Full Professor at EURECOM, Sophia-Antipolis. He is the recipient of the best paper award in IEEE SECON 2008, and IEEE WoWMoM 2012.
\end{IEEEbiography} 

\end{document}